\documentclass[11pt]{article}
\usepackage[margin=1in]{geometry}
\usepackage{caption}
\usepackage{subcaption}
\usepackage{import, comment}
\usepackage{hyperref}

\usepackage{multirow}
\usepackage{booktabs}
\usepackage{tabularx}
\usepackage{setspace} 

\usepackage{bm,amsmath,amssymb,amsthm,amsfonts,mathrsfs,graphicx,caption,bbold,thm-restate,multirow,xspace,ifthen,color,enumitem,algorithm, algpseudocode}
\usepackage{thm-restate, csquotes}

\theoremstyle{plain}

\newtheorem{lemma}{Lemma}

\theoremstyle{definition}
\newtheorem{definition}{Definition}

\theoremstyle{remark}

\usepackage{bm}
\usepackage{algorithm, algpseudocode}
\usepackage{thm-restate, csquotes}

\usepackage{natbib}
 \bibpunct[, ]{(}{)}{,}{a}{}{,}

\newcommand{\E}{\mathbb{E}}

\newcommand{\pay}{p_{j,n+1}}
\newcommand{\exppay}{\bar{p}_{j}}
\newcommand{\alloc}{x_{j,n+1}}
\newcommand{\expalloc}{\bar{x}_{j}}
\newcommand{\Acal}{\mathcal{A}}
\newcommand{\bbf}{\mathbf{b}}
\newcommand{\Fcal}{\mathcal{F}}
\newcommand{\hatsigma}{\hat{\sigma}}
\newcommand{\realpay}{{\sum_{j\in[J]}\hat{p}_{j,t}}}

\title{Multi-Platform Budget Management in Ad Markets with Non-IC Auctions}
\author{
  \normalsize Fransisca Susan\\
  \normalsize MIT Operations Research Center\\
  \normalsize \texttt{fsusan@mit.edu}
  \and
  \normalsize Negin Golrezaei\\
  \normalsize MIT Sloan School of Management\\
  \normalsize \texttt{golrezae@mit.edu}
  \and
  \normalsize Okke Schrijvers\\
  \normalsize Meta Platforms, Inc.\\
  \normalsize \texttt{okke@meta.com}
}
\date{}

\begin{document}
\maketitle
\small{In online advertising markets, budget-constrained advertisers acquire ad placements through repeated bidding in auctions on various platforms. We present a strategy for bidding optimally in a set of auctions that may or may not be incentive-compatible under the presence of budget constraints. Our strategy maximizes the expected total utility across auctions while satisfying the advertiser's budget constraints in expectation. Additionally, we investigate the online setting where the advertiser must submit bids across platforms while learning about other bidders' bids over time. Our algorithm has $O(T^{3/4})$ regret under the full-information setting. Finally, we demonstrate that our algorithms have superior cumulative regret on both synthetic and real-world datasets of ad placement auctions, compared to existing adaptive pacing algorithms.}

\normalsize

\section{Introduction}
\label{sec:intro}
The online advertising industry has seen tremendous growth in recent years, with an estimated worth of $190$ billion dollars in the United States alone in 2021 \cite{statista2021}. Centralized platforms like Google and Meta play a significant role in this thriving industry, as advertisers, including small businesses and marketing practitioners, compete daily in thousands of auctions for valuable ad impressions. With multiple channels now available, advertisers are faced with the complex task of managing their budgets dynamically across different platforms, each with its own auction format that might or might not be incentive compatible.

Advertisers aim to maximize their utility (such as exposure, impressions, CTRs, etc.) while staying within their budget limits. One of the challenges that they have is that each platform utilizes its own unique auction format that requires a tailored strategy. The problem becomes harder because advertisers have to balance their current and future bidding opportunities across platforms to ensure their budget lasts throughout the duration of their campaign, as running out of funds before the end of the campaign could result in significant losses. Moreover, they also might start with little/no prior information about how their competitors bid in each platform, adding another ambiguity to their problem. 

In the area of budget management in ad markets, previous work has looked at the problem of dividing budgets across standard auctions with contextual information on users and items \cite{balseiro2022contextual}. While this study provides a valuable value-pacing strategy and proof of Bayes-Nash equilibrium, it is limited to standard auctions that exclude common position auctions such as generalized second price (GSP), or generalized
first price (GFP) auctions, as well as being applicable only in offline settings. On the other hand, \cite{balseiro2019learning} focuses on the dynamic setting and provides an adaptive pacing strategy to divide the advertiser's budget over time, but only considers second price auctions that are incentive compatible. This study narrows its scope to a single type of auction, instead of the broader scope of multi-platform auctions. In contrast, our work expands upon the prior research by defining a broader space of auctions that includes GSP and GFP and examining the optimal value-pacing strategy for this larger class of auctions. We consider multiple auctions on different platforms, both in the offline and dynamic settings, providing a comprehensive solution for budget management in ad markets. We answer the following key research question:

\begin{displayquote}
What is the optimal budget allocation strategy for an advertiser in a multi-platform advertising market with varying and potentially non-incentive compatible auction formats? How can this strategy be translated and optimized to maximize the utility in a dynamic setting with limited prior information about competitors' bidding behavior?
\end{displayquote}

To answer these questions, we consider multiple platforms that each runs its own auction format, which may or may not be incentive compatible. We look at both the offline and online settings. In the offline setting, the advertiser knows her budget, the format of each auction, her realized values in each auction, and the distribution of other bids in each auction, but is unaware of the actual bids made by other bidders. The goal is to optimize bidding strategies across different platforms, leading to an optimal expenditure in each of them. In the online setting, the advertiser knows the format of each auction and her total budget at the outset, and dynamically learns over time the best way to scale her values and stay within budget, given the (unknown) distribution of other bids in each platform. We now summarize our main contributions.

\textbf{Main Contributions.} 
\textit{Offline setting: Optimal budget allocation strategy across multiple platforms and auctions.} For the offline setting, under mild assumptions, we present a value-pacing strategy for allocating an advertiser's budget across multiple platforms and auctions, which may not be incentive-compatible, with the goal of maximizing the advertiser's expected (quasi-linear) utility. The strategy involves determining the optimal pacing value for each auction based on the budget, and finding the best response strategy for each auction format. This study extends previous work in \cite{balseiro2022contextual}, which introduced a value-pacing equilibrium concepts in a complete information and contextual environment, but only focused on standard auctions\footnote{auctions in which the highest bidder wins, which includes first price and second price auctions}. We expand upon this by considering a more general auction format that may not be incentive-compatible and show that our results still hold under certain assumptions.

\textit{Dynamic setting: Adaptive budget optimization strategy over time and across platforms.} We present a gradient descent-based strategy that adapts by continually learning the pacing value and the best response strategy for each auction format over time. We build upon prior work in \cite{balseiro2019learning}, which only examined the second price auction. Our algorithm obtains vanishing regret theoretically.

\textit{Numerical studies on real and synthetic auction dataset.} We present an empirical evaluation of our online algorithm on both synthetic and real dataset, and demonstrate that our algorithm outperforms the benchmark in all settings. Despite the fact that many of the theoretical assumptions stated in our theorems are not satisfied in both real and synthetic data, our results still reveal vanishing regrets.

\section{Other Related Works}
\label{sec:relatedwork}

\textbf{Auctions with financially constrained buyers.} The study of auctions with financially constrained buyers has a long history, e.g. \cite{laffont1996optimal, mehta2007adwords, abrams2008comprehensive, azar2009gsp, goel2010gsp, karande2013pacing, rusmevichientong2006, feldman2007budget, hosanagar2008, conitzer2021pacing, conitzer2022pacing}.
In the early work by \citet{laffont1996optimal}, the authors considered a symmetric setting with identical budgets for all buyers and showed that an all-pay auction with reserves is the optimal mechanism. Since then, a large number of works have extended and improved upon this result in various directions. For example, Golrezaei et al. \cite{golrezaei2021bidding} designed the optimal auction for ROI-constrained buyers and showed through empirical evidence that some buyers in online advertising are indeed budget constrained. Pai et al. \cite{pai2014optimal} considered a multi-dimensional setting where both valuations and budgets were private information. In another work, Golrezaei et al. \cite{golrezaei2021bidding} looked at budget and ROI-constrained buyers in repeated second price auctions and presented a simple optimal policy for buyers to bid optimally. They also extended this policy to an adaptive algorithm for dynamic settings, and from the seller's perspective, they presented a pricing algorithm that maximizes revenue.

\textbf{Learning in sequential auctions with budget constraints.} Several works have explored the topic of learning in auctions with budget constraints in various settings of repeated auctions. \cite{han2020optimal} presents an optimal no-regret learning algorithm in repeated first price auctions using the monotone group contextual bandit approach. For repeated second price auctions, \cite{balseiro2019learning} proposes an optimal gradient-based algorithm that paces the values at each round according to budget constraint. Another work, \cite{kumar2022optimal} considers both an episodic setting that is stationary for each episode, but can be vastly different across episodes, and a non-stationary setting where both prices and values change over time. They extend the algorithm in \cite{balseiro2019learning} to vary the approximate spend rates over time, and present an algorithm that approximates the optimal bidding strategy up to a factor in the non-stationary setting. \footnote{See also \cite{MorgensternR16, roughgarden2016ironing, mohri2016learning, CaiD17, DudikHLSSV17,kanoria2014dynamic, golrezaei2019IC, golrezaei2018dynamic, golrezaei2021bidding} for works that use online learning algorithms for optimizing reserve prices in auctions.}

\textbf{Multi-channel auction problem.} Our work also relates to other papers that look at optimization across auctions on multiple channels, e.g. \cite{avadhanula2021} \cite{deng2023multi}, \cite{liaw2022efficiency}. \cite{avadhanula2021} models an online problem to divide budget across multiple channels/platforms over time as a stochastic bandit with knapsacks problem, and presents a no-regret algorithm for both the discrete and continuous bid-spaces. Another prominent line of work in this area is the study of autobidding, a mechanism where each channel/platform can bid on behalf of the advertiser through an automated black box algorithm - that was first introduced in \cite{aggarwal2019autobidding}, and has since been examined through other lenses such as welfare efficiency \cite{deng2022efficiency} and fairness \cite{deng2022fairness}. 

\section{Model}
\label{sec:model}
\textbf{Notation.} We use the following notations throughout the paper.
For a set of bids $\bbf = (b_1,\ldots,b_{n+1})$, $\bbf_{-i}$ denotes the set of all bids except for the $i^{\text{th}}$ bid, $b_i$.

\subsection{Offline Setting}
\textbf{Multi-platform setting.} Suppose that there are $J$ platforms, where platform $j\in[J]$ runs a single-unit or multi-unit auction $\Acal_j$. Here, $\Acal_j$ can be, for example, Vickrey-Clarke-Groves (VCG), GSP, or GFP auctions. Suppose that there are $(n+1)$ total bidders who are permitted to participate in all $J$ platforms. The values of all bidders for all auctions are restricted within a bounded space $\mathcal{V}$, where we assume that there exists a constant $U$ such that $\mathcal{V}\subseteq(0, U)$. 

\textbf{Auction rules.} We consider a setting where each bidder submits a scalar bid to each auction (platform)\footnote{If a bidder chooses not to participate in an auction, they can submit a bid of $0$.}, and for each auction, given all the bids $\mathbf{b}_j = (b_{j,1},b_{j,2},\ldots,b_{{j,1n+1}})$, we have the allocation and payment rules that map the bid vector $\mathbf{b}_j$ to the corresponding click-through rate (CTR) and payment for each bidder. The CTR is dependent on the position achieved by each bidder, with each position corresponding to a CTR value in $[0,1]$ and no position resulting in a CTR of 0. Specifically, for auction $j$, we have the allocation function $x_j:\mathbb{R}^{n+1}\rightarrow[0,1]^{n+1}$ that maps the bid vector $\mathbf{b}_j$ to the CTR vector, and the payment function $p_j:\mathbb{R}^{n+1}\rightarrow\mathbb{R}^{n+1}$ that maps the bid vector $\mathbf{b}_j$ to the payment vector.

We focus on the perspective of a specific bidder, the $(n+1)^{\text{th}}$ bidder, whose value in auction $j$, $v_j$, is drawn from a bounded distribution $V_j$. In auction $j$, it is assumed that the bids of all other bidders are identically distributed and drawn from a bounded distribution $F_j$, whose bids are in $[0,U]$. By taking expectation with respect to the distribution of other bids, for each auction $j$, we can derive the expected  allocation and payment of the $(n+1)^{\text{th}}$ bidder, as functions of her own bid. Specifically, $$\expalloc(b) := \E_{\{b_{j,i}\}_{i\in[n]}\sim F_j}[\alloc(b_{j,1},\ldots, b_{j,n}, b)]$$ represents the expected allocation and \[\exppay(b) := \E_{\{b_{j,i}\}_{i\in[n]}\sim F_j}[\pay(b_{j,1},\ldots, b_{j,n}, b)]\] represents the expected payment of the last bidder when she submits a bid $b$, where $x_{j,n+1}$ and $p_{j,n+1}$ are the $(n+1)^{\text{th}}$ element of the vectors resulted from the functions $x_j$ and $p_j$, respectively. For simplicity, we shorten these expressions to $\expalloc(b)$ and $\exppay(b)$. Throughout the paper, we assume that both functions are continuously differentiable almost everywhere with respect to the last bidder's bid $b$. These assumptions are implied when the distributions of other bids and the distribution of values are continuous, which is a common occurrence in the literature (\cite{balseiro2017budget, balseiro2022contextual, brustle2022price}).

\textbf{Offline bidding problem.} We aim to determine an optimal bidding strategy for the $(n+1)^{\text{th}}$ bidder who wants to divide her budget across auctions to maximize her total expected utility. The bidder is aware of her budget $\rho$, the format of each auction $(\Acal_1,\ldots,\Acal_J)$, her value distribution in each auction $\mathbf{V} = (V_1,\ldots,V_J)$, and the distribution of other bids in each auction $\mathcal{F} =(F_1,\ldots,F_J)$, but she does not know the actual bids made by other bidders. Moreover, her expected utility in auction $j$ when her value and bid are respectively  $v_j$  and $b$
 is the difference between the expected gain from being allocated and the expected price paid, i.e. $v_j\expalloc(b) - \exppay(b)$.

We want to determine the bidding strategy for bidder n+1, $\{b_{j}(\cdot)\}_{j=1}^J$ as a function of her values in all platforms and her budget. This is challenging because the strategies across channels cannot be optimized individually due to the existence of the budget constraint. In fact, they need to be optimized jointly, creating a big decision space. For simplicity, we assume that the dependence on the values on other platforms and the budget is implicit, and we use $b_j(v_j)$ as the bid of bidder $n+1$ (possibly random) for platform $j$ (auction $\Acal_j$) when the last bidder's value in auction $j$ is $v_j$. Hence, we want to find a set of bidding functions  $\{b_j(\cdot)\}_{j=1}^J$ for the last bidder that maximizes her total expected utility across auctions, $\sum_{j=1}^Jv_j\expalloc(b_j(v_j)) - \exppay(b_j(v_j))$, such that her expected expenditure does not exceed $\rho$. Precisely,  we aim to solve the following optimization:
\begin{equation}
\label{eq:offlineoptproblem}
\begin{aligned}
Z(\Fcal, \mathbf{V}, \rho)
&:= \underset{\left\{b_j(\cdot):\mathcal{V}\rightarrow\mathbb{R}_{\geq 0}\right\}_{j=1}^J}{\max}\quad \sum_{j=1}^J\E_{\Fcal, \mathbf{V}}\left[v_j\expalloc(b_j(v_j))-\exppay(b_j(v_j))\right]\\
&s.t.\quad \sum_{j=1}^J\E_{\Fcal, \mathbf{V}}\left[\exppay(b_j(v_j))\right]\leq \rho.
\end{aligned}
\end{equation}
Here, for simplicity, we shorten the notation $b_{j,n+1}$ (bidder n+1's bid) to $b_j$ instead, which is a function that maps bidder n+1's value to a bid. Hence, $\rho$ is the per-round budget, $(b_1(\cdot),\ldots,b_J(\cdot))$ is the set of bidding strategies, $\Fcal = (F_1,\ldots,F_J)$ is the family of distribution of other bids, such that the $j^{\text{th}}$ element corresponds to auction $\Acal_j$, and $\mathbf{V} =(V_1, \ldots, V_J)$ is the set of valuation distributions (of the last bidder) across all the platforms.

\subsection{Online/Repeated Setting}
In the online problem, we assume that $J$ auctions run simultaneously at each round. Then,  the last bidder, knowing only her realized values in each auction, has to submit her bids in the current auction without knowing her values in all future auctions. The bidder also knows her average budget per round $\rho$ and the total number of rounds $T$ in the beginning\footnote{The assumption of knowing $T$ can be relaxed using the doubling trick.}. She has to bid in each auction on each round, while learning the distribution of her competitors. We consider a full-information setting where all the other bids are revealed in each round for each auction, which is commonly studied in the literature (\cite{tran2014efficient, chen2022dynamic, conitzer2022multiplicative}). Letting $\bm{X} = (\{v_{1,t},\ldots,v_{J,t}\}_{t\in[T]}, (F_1,\ldots,F_J))$, we define the in-hindsight benchmark as
\begin{equation}
\begin{aligned}
\label{eq:generalonlinebenchmark}
OPT(\bm{X}, \rho)
&:=\underset{\substack{\left\{b_{j,t}(\cdot):\mathcal{V}\rightarrow\mathbb{R}_{\geq 0}\right\}\\t\in[T],j\in[J]}}{\max}\quad \sum_{t=1}^T\sum_{j=1}^J v_{j,t}\expalloc(b_{j,t}(v_{j,t})) - \exppay(b_{j,t}(v_{j,t}))\\
&s.t.\quad \sum_{t=1}^T\sum_{j=1}^J\exppay(b_{j,t}(v_{j,t}))\leq \rho T,
\end{aligned}
\end{equation}
where $v_{j,t}$ and $b_{j,t}$ are the bidder's value and bid at round $t$, respectively. The following proposition relates $OPT(\bm{X}, \rho) $ to the offline benchmark $Z(\Fcal, \mathbf{V}, \rho)$, defined in the previous section. The following proposition, which is proved in Section~\ref{sec:prop1proof}, shows that $T\cdot Z(\Fcal, \mathbf{V}, \rho)$ is a stronger benchmark than $\E[OPT(\bm{X}, \rho)]$. 

\begin{restatable}{prop}{offlinetoonline}
\label{prop:offlinetoonline}
The optimal solution to the  problem in Equation \eqref{eq:offlineoptproblem} provides an upper-bound for $\frac{1}{T}\cdot\E[OPT(\bm{X}, \rho)]$; that is, we have  
$\E[OPT(\bm{X}, \rho)]\le T\cdot Z(\Fcal, \mathbf{V}, \rho)$. 
\end{restatable}

A strategy $\beta$ with bidding strategy $b_{j,t}^{\beta}(\cdot)$ aims to maximize the expected utility and satisfies the budget constraint: 
\begin{equation}
\label{eq:onlinebudgetconstraint}
\begin{aligned}
    \sum_{t=1}^T\sum_{j=1}^J\E_{V_j, F_j}\left[\pay\left(b_{j,t}^\beta(v_{j,t})\right)\right]\leq\rho T.
\end{aligned}
\end{equation}
The bidding strategy $b_{j,t}^{\beta}(\cdot)$ maps the last bidder's budget, her realized values in all auctions up to round $t$, and all of the other bidders' bids up to round $t-1$ to a set of bids (one for each auction) at round $t$.

We measure the performance of any strategy $\beta$ that satisfies the budget constraint in expectation against the in-hindsight benchmark with the following cumulative regret:
\begin{equation}
\label{eq:avgregret}
\begin{aligned}
    \text{Reg}^\beta(T,\rho) &= \E\Bigg[OPT(\bm{X}, \rho)- \left(\sum_{t=1}^T\sum_{j=1}^J v_{j,t}\expalloc\left(b_{j,t}^{\beta}(v_{j,t})\right) - \exppay\left(b_{j,t}^{\beta}(v_{j,t})\right)\right)\Bigg],
\end{aligned}
\end{equation}
where the expectation is taken with respect to the distribution of the last bidder's values.

\section{Optimal Solution for the Offline Setting}
\label{sec:offline}
In this section, we show that the \emph{Value-Pacing strategy} is optimal in maximizing the expected utility across auctions while adhering to the budget constraint in expectation.  We make some assumptions to ensure the optimality of the \emph{Value-Pacing strategy}. First, we assume that the expected allocation, $\expalloc(b)$, and the expected payment, $\exppay(b)$, are continuously differentiable with respect to $b$ for any $b\in\mathbb{R}$ where the expectations are taken with respect to the distribution of other bids. Second, we assume that for any values $v\in\mathcal{V}$ and any auction $j$, the best response function $\sigma_j:\mathcal{V}\rightarrow\mathbb{R}$, which maps a value $v$ to a corresponding bid $b^*$ that maximizes the expected utility $v\expalloc(b) - \exppay(b)$ (defined in details in Definition~\ref{def:bestresponse}), is unique and continuously differentiable almost everywhere in $(0,\infty)$ with respect to $v$. The first set of assumptions is both natural, for example, it is satisfied when the distributions of other bids are all continuously differentiable for both GFP and GSP auctions, as demonstrated in Section \ref{sec:sufficient}. The second assumption is satisfied in an incentive-compatible auction, when the distribution of other bids is continuously differentiable and the position-wise discount factor decreases exponentially in GFP, and is implied by many previous settings in the literature (\cite{balseiro2022contextual}, \cite{balseiro2019learning}), see Section~\ref{sec:sufficient} for some sufficient conditions under which the both assumption holds.

A Value-Pacing strategy starts by pacing the value in each auction through scaling down values with a constant factor, which is determined by the budget constraint. The best response function is then applied to the paced values to compute the  bid. Before describing the strategy in details, we define the best response function as follows. 
\begin{definition}
\label{def:bestresponse}
$\sigma_j:\mathcal{V}\rightarrow\mathbb{R}$ is the best response function in auction $j$ if
\begin{equation*}
    \sigma_j(v)\in\arg\max_{b\in\mathbb{R}}~v\expalloc(b) - \exppay(b).
\end{equation*}
\end{definition}
Note that we always have $v\expalloc(b) - \exppay(b)\leq v\cdot1 - 0\leq U$. Furthermore, the function $(v\expalloc(b) - \exppay(b))$ is differentiable with respect to $b$ since both the expected allocation and expenditure functions are. So, when the bid value  is optimal, based on the first principle, we have 
\begin{equation}
\label{eq:sigmaderiv}
\begin{aligned}
v\frac{d\expalloc}{db}(\sigma_j(v)) - \frac{d\exppay}{db}(\sigma_j(v)) &= 0.
\end{aligned}
\end{equation}
Here, the first order condition applies to the unconstrained optimization over $b\in\mathbb{R}$. Note that we can assume that $\sigma_j(v)> 0$ when $v>0$ because the best response function $\sigma_j$ is unique for $v\in(0,\infty)$ and for all auctions that we consider, it is reasonable to assume that if we bid a non-positive number, the total expected utility is $0$. We also define the Lagrangian dual of the optimization problem in Equation~\eqref{eq:offlineoptproblem} as follows:
\begin{equation}
\label{eq:dualofflineproblem}
\begin{aligned}
&\min_{\mu\geq 0}~q(\mu,\mathbf{V},{\Fcal}, \rho) \\
= &\min_{\mu\geq 0}~\E\left[\underset{\left\{b_j:\mathcal{V}\rightarrow\mathbb{R}_{\geq 0}\right\}_{j=1}^J}{\max}~\sum_{j=1}^J(v_j\expalloc(b_j)-(1+\mu)\exppay(b_j))\right] +\mu\rho\\
= &\min_{\mu\geq 0}~\E\left[\underset{\left\{b_j:\mathcal{V}\rightarrow\mathbb{R}_{\geq 0}\right\}_{j=1}^J}{\max}~(1+\mu)\sum_{j=1}^J\left(\frac{v_j}{1+\mu}\expalloc(b_j)-\exppay(b_j)\right)\right] + \mu\rho.
\end{aligned}
\end{equation}

We want to show that the Value-Pacing strategy, where we bid $\sigma_j\left(\frac{v_j}{1+\mu^*}\right)$ in auction $j$ for some pacing factor $\mu^*$, is optimal for Problem \eqref{eq:offlineoptproblem}. Specifically, when $\mu^*$ minimizes the dual of the optimization problem, i.e. $\mu^*\in\arg\min_{\mu\ge 0}q(\mu,\mathbf{V},{\Fcal}, \rho)$, the strategy $\{b_j(v_j) = \sigma_j({v_j}/{(1+\mu^*)})\}_{j\in[J]}$ is optimal.

The dual function can be simplified by recognizing that for a fixed value of $\mu$, the optimal set of bids that obtains $q(\mu,\mathbf{V},{\Fcal}, \rho)$ is the set of best response bids, where for auction $j$, the best response function is applied to a scaled value of $v_j/(1+\mu)$. In other words, bidding optimally as if the "real" value for auction $j$ is $v_j/(1+\mu)$ maximizes the dual in Equation~\eqref{eq:dualofflineproblem} for each dual multiplier $\mu$. This is summarized and proven in the following lemma.

\begin{lemma}
\label{lemma:pacingmany}
For a bidder with value $(v_1,\ldots,v_J)\in\mathcal{V}^J$ and budget $\rho$, and for some dual multiplier $\mu\geq 0$,
\begin{equation*}
    b_j(v_j) = \sigma_j\left(\frac{v_j}{1+\mu}\right),~j\in\{1,\ldots,J\}
\end{equation*}
maximizes
\begin{equation*}
    (1+\mu)\sum_{j=1}^J\left(\frac{v_j}{1+\mu}\expalloc(b_j)-\exppay(b_j)\right) + \mu\rho.
\end{equation*}
\end{lemma}
\begin{proof}
For each $j\in\{1,\ldots,J\}$, by the definition of $\sigma_j$, we have 
\begin{equation*}
\begin{aligned}
\frac{v_j}{1+\mu}\expalloc\left(\sigma_j\left(\frac{v_j}{1+\mu}\right)\right) - \exppay\left(\sigma_j\left(\frac{v_j}{1+\mu}\right)\right)\geq \left(\frac{v_j}{1+\mu}\expalloc\left(b\right)\right)-\exppay\left(b\right),
\end{aligned}
\end{equation*}
for any $b\in\mathbb{R}$, which proves the lemma, since all the other terms only depend on $\mu$ and $\rho$ (both fixed here).
\end{proof}

We can now redefine $q(\mu,\mathbf{V},{\Fcal}, \rho)$ by substituting the optimal set of bids.
\begin{equation}
\label{eq:substituteddualmany}
\begin{aligned}
&q(\mu,\mathbf{V},{\Fcal}, \rho) \\
&= {\mu}\rho + \E\left[(1+\mu)\sum_{j=1}^J \frac{v_j}{1+\mu}\expalloc\left(\sigma_j\left(\frac{v_j}{1+\mu}\right)\right)-\exppay\left(\sigma_j\left(\frac{v_j}{1+\mu}\right)\right)\right]\\
&= {\mu}\rho + \E\left[\sum_{j=1}^J {v_j}\expalloc\left(\sigma_j\left(\frac{v_j}{1+\mu}\right)\right)-(1+\mu)\exppay\left(\sigma_j\left(\frac{v_j}{1+\mu}\right)\right)\right].
\end{aligned}
\end{equation}

We next present the main result in the offline setting, which is proved in Section~\ref{sec:offlineproof}.

\begin{restatable}{theorem}{offlinethm}
\label{thm:offline}
Suppose that for each auction, 1) the expected allocation and expenditure functions are continuously differentiable with respect to the last bidder's bid\footnote{Again, the expectations are taken with respect to the other bids.}, 2) the best response function is unique and continuously differentiable almost everywhere in $(0,\infty)$ with respect to the last bidder's value, and 3) the expected utility of bidding any values less than or equal to the realized values in any auction is always non-negative. If $\mu^*\in\arg\min_{{\mu}\geq 0} q(\mu,\mathbf{V},{\Fcal}, \rho)$, then $b_j = \sigma_j(v_j/(1+\mu^*))$ for all $j=1,\ldots,J$, maximizes Problem~\eqref{eq:offlineoptproblem}, $\sum_{j=1}^J\E\left[v_j\expalloc(b_j(v_j))-\exppay(b_j(v_j))\right])$, over $\bbf\in\mathbb{R}^J$, and the value obtained is $Z(\Fcal, \mathbf{V}, \rho)$.
\end{restatable}
It is important to note that the third assumption is satisfied in all VCG, GFP, and GSP auctions.

\section{Online Optimal Bidding Strategy Across Auctions}
\label{sec:online}
In this section, we present the \emph{Adaptive-Value-Pacing} algorithm, which is asymptotically optimal for the online bidding problem. In the online setting, the last bidder has access to the length of the horizon $T$ (i.e., the length of her campaign), her per round budget $\rho$, and the set of auction formats $(\Acal_1,\ldots,\Acal_J)$ from the start. At each round, all $J$ auctions run simultaneously, and the bidder must submit her bids (simultaneously to $J$ auctions based on her realized values). She then receives all other bids as feedback (we consider the full-information setting), together with the allocation and payment she gets in each auction at that round. We assume that the distribution of bids from other bidders in the same auction is identical, as in the offline setting. Thus, in auction $j$, the bids of all other bidders are identically distributed and drawn from a bounded distribution $F_j$.

The Adaptive-Value-Pacing (AVP) algorithm works by finding the optimal pacing multiplier and obtaining accurate estimates of the cumulative distribution functions of the other bidders' distributions in each auction. At each round $t$, the algorithm calculates a pacing multiplier $\mu_t$ and scales the values in each auction by multiplying them with $1/(1+\mu_t)$. The algorithm then submits bids that are estimated to be the best response to the current scaled values, while also ensuring that the total budget has not been exceeded. Upon receiving feedback after each round, the algorithm updates the pacing multiplier using a dual-based gradient update and continually refines its estimates of the best response functions for each auction through an analysis of the collected bid data. Our results, as stated in Theorem~\ref{thm:online}, show that under certain conditions, the Adaptive-Value-Pacing algorithm is asymptotically optimal with regret of $O(T^{3/4})$.

\begin{algorithm}[ht!]
    \caption{Adaptive Value Pacing for Repeated Auctions with Budget Constraint (\textsc{AVP})}
    \label{alg:online}
    \begin{algorithmic}[1]
        \State Input: total period $T$; average budget per round $\rho$; the set of auction formats $(\Acal_1,\ldots,\Acal_J)$; a learning parameter $\epsilon$.
        \State Output: bids in each auction at each period.
        \State Initialize a multiplier $\mu_1$ randomly in $[0, JU/\rho]$ and set the remaining budget to $\tilde{B}_{1} = \rho T$. Set the upper-bound on the multiplier $\bar{\mu} = JU/\rho$, and initialize the estimated best response functions: $\hatsigma_{j,1}(v) = v$ for each auction $j$ any $v\le U$.
        \For{$t = 1,2,\ldots,T$}
            \State \underline{Observing the realized values $\{v_{j,t}\}_{j=1}^J$}
            \State If the remaining budget $\tilde{B}_t< JU$, bid $0$ in each auction. Otherwise, choose a bid for each auction,            
            \begin{align*}
                b_{j,t} = \hatsigma_{j,t}\left(\frac{v_{j,t}}{1+\mu_t}\right),
            \end{align*}
            where
            \[ 
            \hatsigma_{j,t} (v) = \left\{
            \begin{array}{ll}
                  v & t = 1 \\
                  \arg\max_{b\in\mathbb{R}}\E_{\hat{F}_{j,t-1}}[v\alloc(\bbf_{j,-(n+1)},b)- \pay(\bbf_{j,-(n+1)},b)] & \text{otherwise.}
            \end{array} 
            \right. 
            \]
            \State \underline{Observing the bids $\{\bbf_{j,-(n+1),t}\}_{j\in[J]}$, allocation $\{\hat{x}_{j,t}\}_{j=1}^J$ and payment $\{\hat{p}_{j,t}\}_{j=1}^J$}\\
            \underline{for the last bidder} 
            \vspace{0.25cm}
            \State Updating the budget multiplier: $$\mu_{t+1} = \min\left\{\max\left\{0, \mu_{t} - \epsilon(\rho - \sum_{j\in[J]}\hat{p}_{j,t})\right\},\bar{\mu}\right\}.$$
            \State Updating the other bids' distributions: $\forall~j\in[J]$, $\hat{F}_{j,t}$ is the empirical cdf of the set of bid $\{\bbf_{j,-(n+1),i}\}_{i=1,\ldots,t}$\footnotemark.
            \State Set the remaining budget, $\tilde{B}_{t+1} = \tilde{B}_{t} - \sum_{j\in[J]}\hat{p}_{j,t}.$
        \EndFor
    \end{algorithmic}
\end{algorithm}
\footnotetext{Recall that, in auction $j$, each element of $\mathbf{b}_{j,n+1}$ is drawn from $V_j$.}

Prior to presenting Theorem~\ref{thm:online}, let us recall that, according to Proposition~\ref{prop:offlinetoonline}, the expected value $OPT(\bm{X},\rho)$ can be upper-bounded by $T\cdot Z(\Fcal, \mathbf{V},\rho)$. Furthermore, recall that the dual of $Z$ is $q(\mu,\Fcal,\mathbf{V},\rho)$. To simplify the notations, let $\Gamma(\mu) = q(\mu,\Fcal,\mathbf{V},\rho)$ for fixed sets of $(V_1,\ldots,V_J)$, $(F_1,\ldots,F_J)$, and  $\rho$. In the  theorem, we will impose two assumptions on the dual function $\Gamma$ to ensure the converge of Algorithm~\ref{alg:online} as a primal-dual algorithm. {Specifically, we require that the dual function $\Gamma(\mu)$ is strongly convex and thrice differentiable with respect to $\mu$, as inspired by \cite{balseiro2019learning}. Strong convexity is a standard assumption that is needed to guarantee that the pacing multipliers converge, while the thrice differentiability is needed to perform Taylor series expansion on the dual. It has been shown that for first-price auctions, these conditions hold when values and competing bids are independent and absolutely continuous with bounded densities and the valuation density is differentiable \cite{balseiro2019learning}. Similar results hold for GFP and GSP auctions.} We now state the theorem that guarantees the performance of our online algorithm, Algorithm~\ref{alg:online}.
\begin{restatable}{theorem}{onlinethm}
\label{thm:online}
Suppose that for each auction $j$, the last bidder's value $v_{j,t}$ is independently drawn from some stationary distribution $V_j$ such that the dual function $\Gamma(\mu)$ is thrice differentiable (with respect to $\mu$) with bounded derivatives and $\lambda$-strongly convex for some $\lambda>0$. Let $AVP$ be an adaptive pacing strategy as described in Algorithm~\ref{alg:online}, whose step size $\epsilon$ which is a function of $T$ satisfies the following: 1) $0<\epsilon<1/JU$, 2) $\lim_{T\rightarrow\infty}\epsilon = 0$, 3) $\lim_{T\rightarrow\infty}T\epsilon = \infty$, and 4) $\epsilon<1/\lambda$. Furthermore, suppose that for each auction, 1) the expected allocation ($\expalloc$) and payment ($\exppay$) functions are both $L$-Lipschitz for some constant $L> 0$\footnote{A Lipschitz function is continuously differentiable almost everywhere.}, and 2) the best response function is unique and continuously differentiable almost everywhere with respect to the last bidder's value. Then, for the bidding algorithm  $\beta = \text{AVP}$ in Algorithm~\ref{alg:online}, we have \begin{equation*}
\begin{aligned}
    \text{Reg}^\beta(T,\rho) 
    &= \E\Bigg[OPT(\bm{X}, \rho) - \Bigg(\sum_{t=1}^T\sum_{j=1}^J v_{j,t}\expalloc\left(b_{j,t}^{\beta}(v_{j,t})\right)- \exppay\left(b_{j,t}^{\beta}(v_{j,t})\right)\Bigg)\Bigg] \\
    &\leq O\left(T^{3/4}\right),
\end{aligned}
\end{equation*}
which is attained when the learning rate $\epsilon = O(T^{-1/4})$.
\end{restatable}

\begin{proof}
One crucial part to prove Algorithm~\ref{alg:online} is to ensure that the best response estimates for each auction improve over time. We demonstrate this through the application of the DKW inequality to the empirical cumulative distribution function, as shown in Lemma~\ref{lemma:BR-estimate}. This implies that, for all $t\in[T]$ and $j\in[J]$, there exists $\delta_{j,t}>0$ such that $\sup_{v\in\mathbb{R}_{\geq 0}}|\sigma_j(x) - \hat{\sigma}_{j,t}(x)|\leq\delta_{j,t}$. 

With that assumption, the proof consists of seven parts. Firstly, the expected utility in-hindsight is upper-bounded by weak duality. Secondly, the difference between the expected utility and expected expenditure in each round when using empirical CDF, compared to when calculated using the true CDF, is lower-bounded. Thirdly, the performance of the AVP strategy is lower-bounded by discarding the utility of all auctions after budget depletion, and it is shown that the budget is not depleted too soon. Fourthly, a second order Taylor expansion is performed on the expected utility per auction around $\mu^*$ to derive a lower-bound, where the absolute mean error $|\mu_t-\mu^*|$ appears as the first-order term and the mean squared error $(\mu_t-\mu^*)^2$ appears as the second-order term. Lastly, the expected absolute mean errors and mean squared errors, summed over all time periods, are upper-bounded. By combining these steps, the regret of the adaptive strategy is upper-bounded and the proof is concluded. The details of the proof are in Section~\ref{sec:proof-online-thm}. 
\end{proof}

 \section{Experiments}
\label{sec:experiment}
\subsection{Synthetic Experiments}

\textbf{Settings.} In this synthetic experiment, we analyze three position auctions: VCG, GFP, and GSP, with VCG being the only incentive-compatible auction. To standardize our results, we fix the number of auctions $J$ at two, the number of competitors $n$ at five, and the number of positions in each auction at three, as to guarantee that the number of bidders exceeds the number of available positions. Note that varying these parameters does not significantly impact the cumulative regret. We establish the discount factor (CTR value) as (1, 0.5, 0.25) for the first, second, and third positions respectively. We set the number of periods to be $T= 10,000$ to observe the long-term behavior of our algorithm and benchmark that we will define shortly.

We vary the format of auctions, distributions, and budget in this experiment. We consider three sets of auctions: in the first set, we have a VCG and a GFP auction, in the second set, we have two GFP auctions, and in the last set, we have a GSP and a GFP auction. The objective is to examine the performance of Algorithm~\ref{alg:online} across different types of auctions, both IC and non-IC. In each setting, we use two types of bid distributions: a lognormal distribution with $\mu=-0.3466$ and $\sigma=0.8326$, and another lognormal distribution with $\mu=-0.5493$ and $\sigma=1.0481$. These distributions are selected so that they both have a mean of $1$, but different variances (Var), where we note that $\text{Var}=1$ for the first distribution and $\text{Var}=2$ for the latter. This allows us  to see the impact of increased variance in competitor bids on the cumulative regret. Although these distributions are not completely bounded, they both have long tails to the right, as is typical for lognormal distributions. For example, the probability that a value drawn from the first distribution is in the interval $[0,10]$ is at least $0.999$, while the probability that a value drawn from the second distribution is in $[0,15]$ is at least $0.999$. For the value of the last bidder, we use the same distribution as the other bids, but multiply the value drawn from the distribution by a random multiplier in $[1,1.5]$, as bids are generally lower than the values. We also vary the average per round budget, $\rho$, to be in $\{1,3\}$, which are around the mean and median of both distributions. Finally, to assess the stability of the results, we conduct 50 runs for each setting and analyze the mean and standard error across the runs.

\textbf{Benchmark.} We compare our strategy to an adaptive pacing algorithm in \cite{balseiro2019learning} as a benchmark and calculate the regret with respect to the best expected utility in-hindsight. The adaptive pacing algorithm in \cite{balseiro2019learning} paces the values across auctions in each round based on the budget constraint, then bids the paced values in all auctions. However, unlike our algorithm, it does not use the best response function, and their paper only considers second price auctions and use the identity functions to bid.

\textbf{Implementation.} In our experiment, we implement the AVP algorithm with a learning rate of $\epsilon=T^{-1/4}$, which is based on the theoretical learning rate that minimizes regret in the analysis of AVP. For the benchmark, we try two learning rates, $T^{-1/2}$ and $T^{-1/4}$, then pick the learning rate that gives the highest utility for the benchmark, which happens to be $T^{-1/4}$. Additionally, we set the value of $U$ (the upper-bound on values) to 10 and 15 when the value distribution is lognormal with variance 1 and 2, respectively. This is because the values drawn from each distribution have a high probability of being in the range [0, U].

\begin{figure}[ht!]
    \centering
    \includegraphics[width=0.9\textwidth]{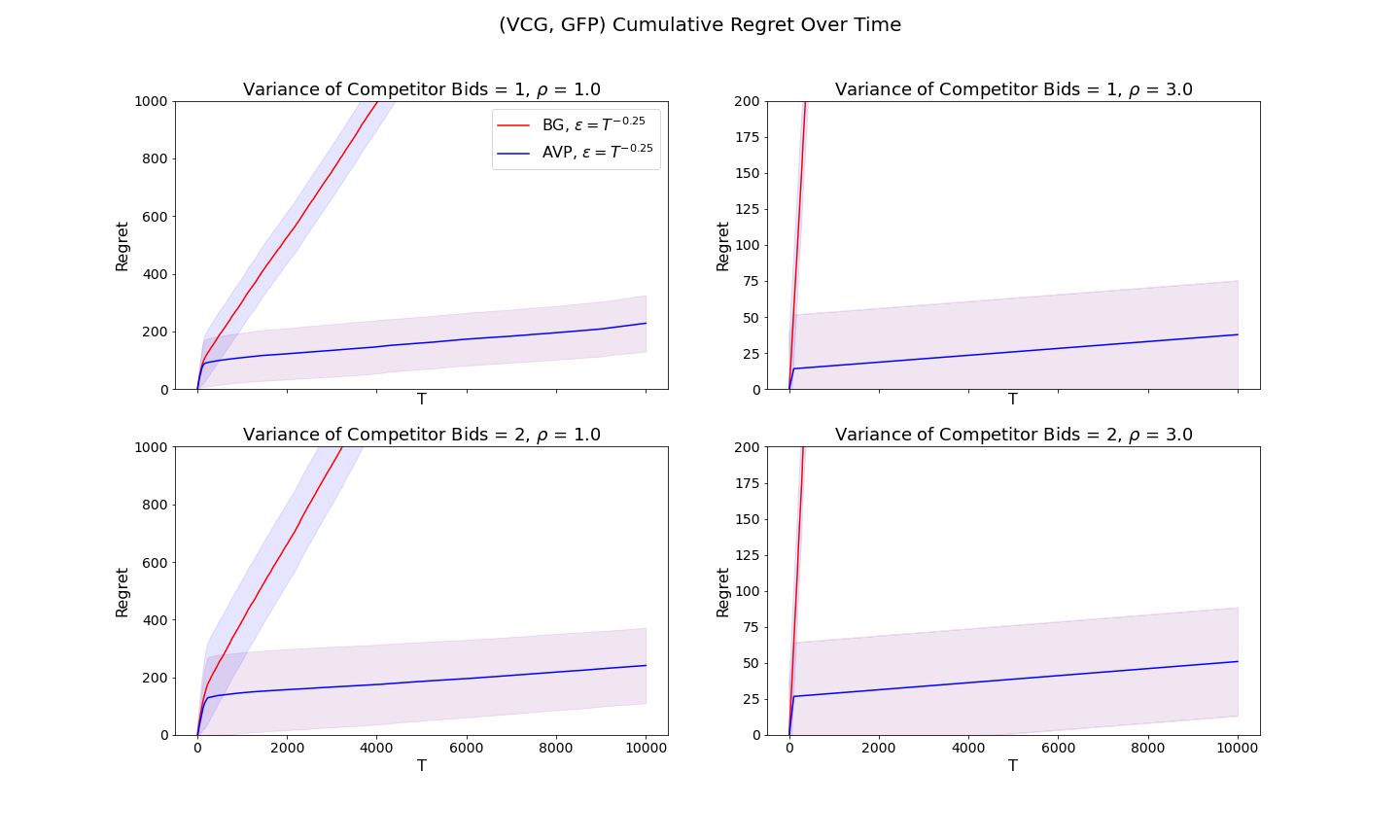}
    \caption{Cumulative regret over time on VCG and GFP Auctions. }
    \label{fig:cumregret_VCG}
\end{figure}

\begin{figure}[ht!]
    \centering
    \includegraphics[width=0.9\textwidth]{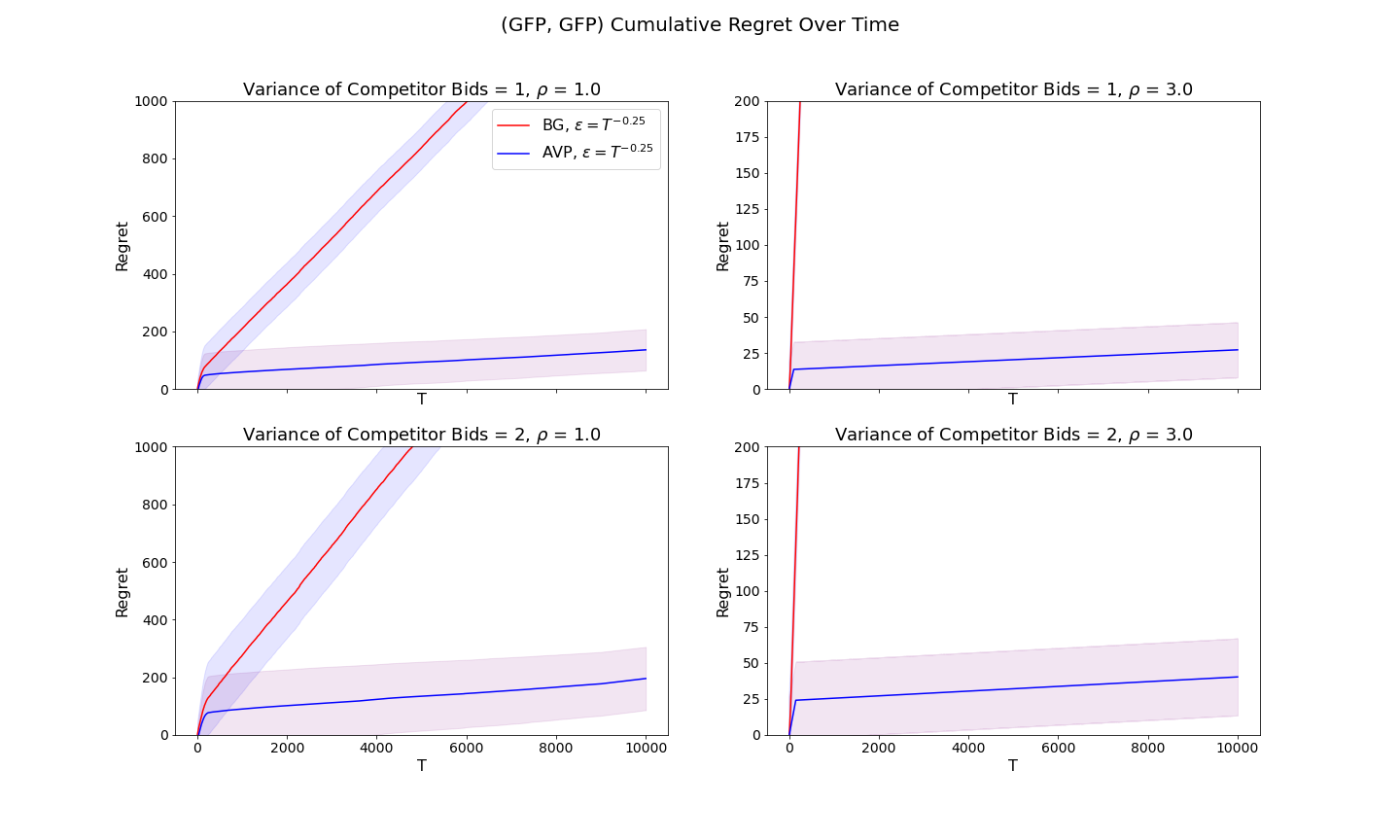}
    \caption{Cumulative regret over time on GFP Auctions.}
    \label{fig:cumregret_GFP}
\end{figure}

\begin{figure}[ht!]
    \centering
    \includegraphics[width=0.9\textwidth]{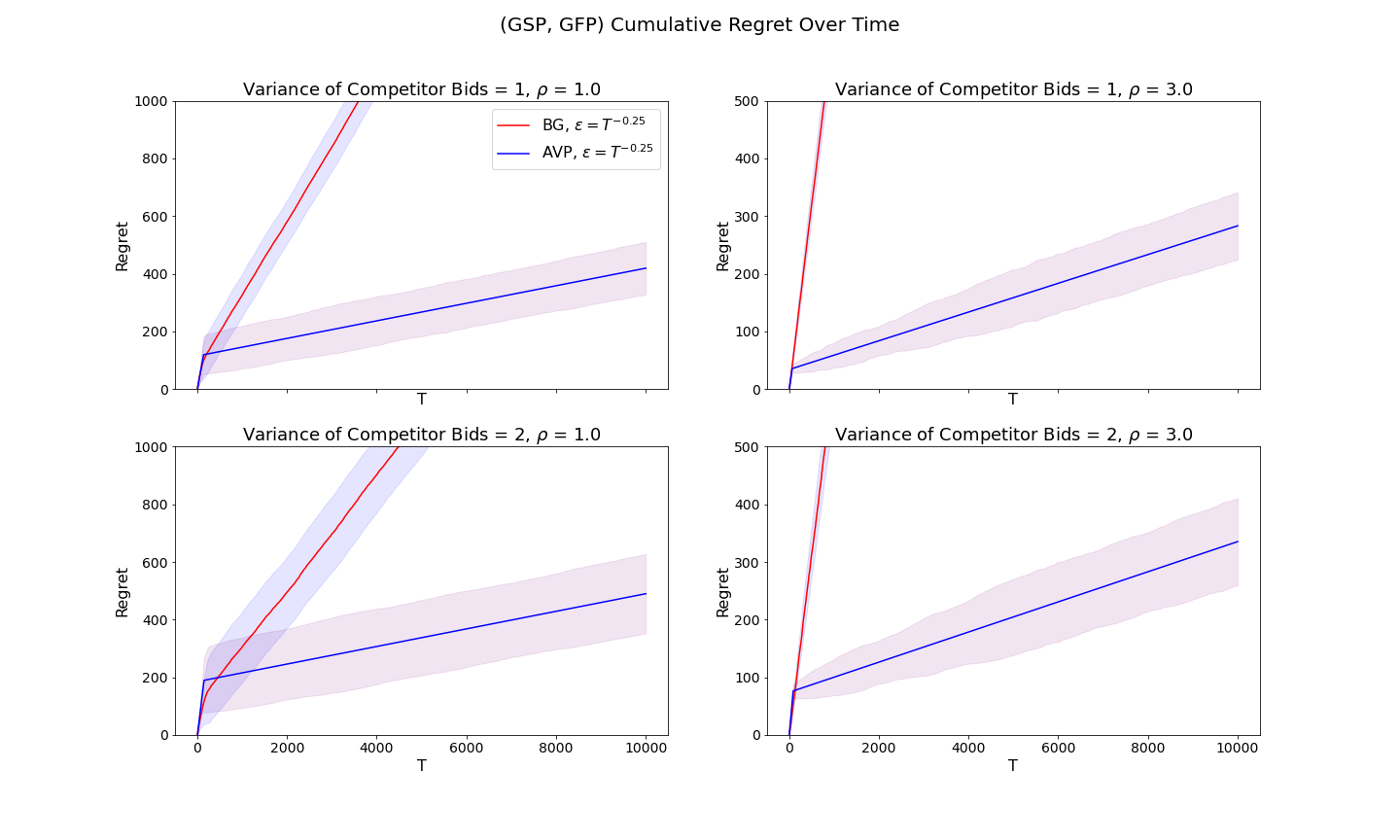}
    \caption{Cumulative regret over time on GSP and GFP Auctions.}
    \label{fig:cumregret_GSP}
\end{figure}

\textbf{Results.} The mean cumulative regret over time for the AVP algorithm and the benchmark algorithm are shown along with the standard errors, obtained over 50 runs, in Figures~\ref{fig:cumregret_VCG}, \ref{fig:cumregret_GFP}, and \ref{fig:cumregret_GSP} for the (VCG,GFP), (GFP,GFP), and (GSP,GFP) settings respectively. Our results demonstrate that the AVP algorithm outperforms the benchmark, which exhibits almost linear regret in all settings. Even in the first setting, which includes a VCG\footnote{For VCG, our algorithm is identical with the benchmark since the best response function in VCG is the identity function, but this is not the case with GFP and GSP.} auction in each round, the AVP algorithm displays significantly better regret compared to the benchmark. We observe that both the AVP algorithm  perform similarly in the first (VCG,GFP) and second (GFP,GFP) settings, indicating that the regret from GFP is likely to dominate the regret from VCG. Moreover, based on the gradients of our cumulative regret, we see that our cumulative regrets are $O(T^{3/4})$, which is the theoretical upper-bound on regret. In the third setting, which includes GSP, the gradients of the cumulative regret of AVP are noticeably larger, compared to the previous settings, although the asymptotic growths are still $O(T^{3/4})$, which is the theoretical upper-bound on the cumulative regret. We can observe this in Figure~\ref{fig:normalized_regret_GSP}, where for AVP (the left figure), when the cumulative regret is normalized by $T^{3/4}$, all the four curves are converging to constant lines, while for the benchmark (the right figure), the normalized cumulative regrets are still increasing. The cumulative regrets are worse in the third setting due to the lack of unique or continuously differentiable best response functions in GSP.

\begin{figure}[ht!]
    \centering
    \includegraphics[width=0.95\textwidth]{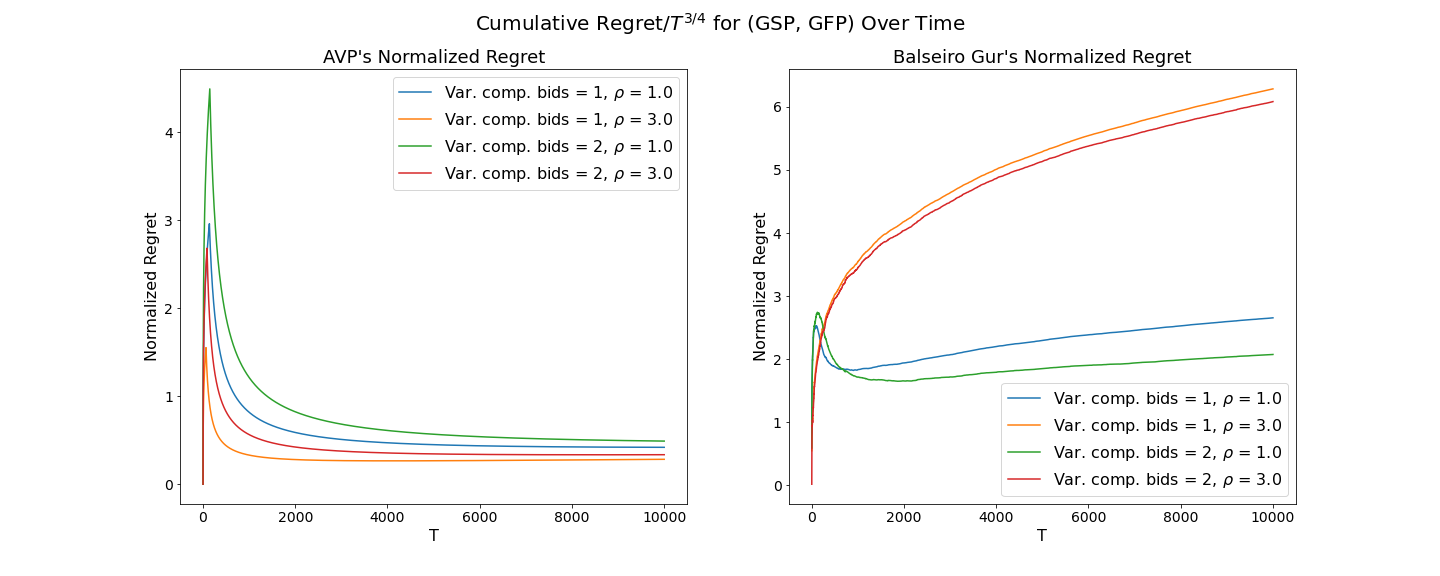}
    \caption{Normalized cumulative regret over time (divided by $T^{3/4}$) on GSP and GFP Auctions, i.e., the normalized regret at time $t$ is the cumulative regret up to time $t$ divided by $t^{3/4}$. See that for AVP (left), the normalized regret curves all converge to constant lines, while for the benchmark (right), the normalized regret curves are not converging.}
    \label{fig:normalized_regret_GSP}
\end{figure}

In terms of variance, we observe that the cumulative regret produced by the AVP algorithm has smaller variance compared to the benchmark in all cases (although this may not be evident in the plots due to the relatively large mean values of the benchmark cumulative regret). The variance of the cumulative regret decreases as the budget per round increases and vice versa in all settings. Moreover, as expected, the cumulative regret variance is higher when the variance of competitor bids is higher, and vice versa.

Finally, it is important to note that the assumptions of our theorem are not fully satisfied by our settings. For instance, the values are not fully bounded, and some best response functions are not unique or have continuous derivatives. Despite this, our empirical results indicate that the cumulative regret is $O(T^{3/4})$.

\subsection{Real Data Experiments}
We construct a real-world dataset obtained from log of a large online internet advertising company. The bids that were logged are normalized to obtain a representative distribution of bids.\footnote{The auctions themselves are simulated in the same manner as for synthetic data. While this means that the bidding data that is used is not directly for the auction formats that are studied, we believe the data is representative enough to understand the effects of different auction formats on the optimal algorithm.} Similar to the synthetic setting, the value of the last bidder is drawn from the same distribution as the competing bids, but the value is multiplied by a random multiplier in the range of $[1, 1.5]$ to account for bidders commonly shading their values. We study the setting of two GFP auctions, each with five competitors and discount factors of $(1, 0.5, 0.25)$, similar to the synthetic case. Furthermore, the algorithms, benchmarks, and implementations are the same as those used in the synthetic case.

\begin{figure}[ht!]
    \centering
    \includegraphics[width=0.8\textwidth]{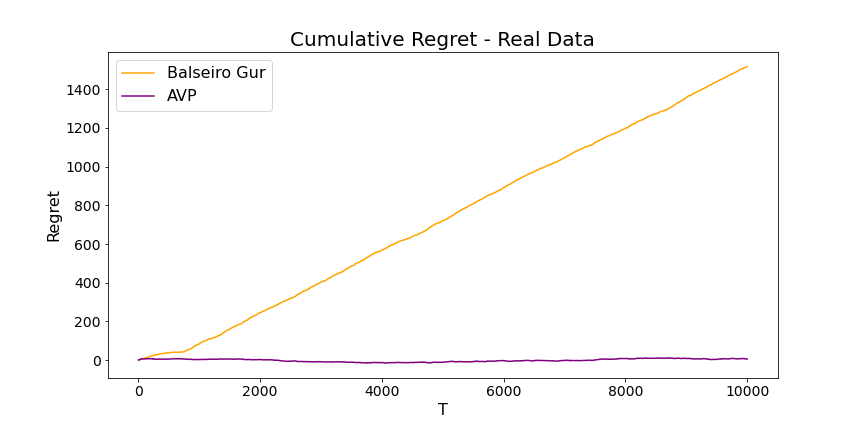}
    \caption{Cumulative regret for the (GFP,GFP) setting on real data. }
    \label{fig:regret_real}
\end{figure}

\textbf{Results.} We plot the cumulative regret of AVP and the benchmark algorithm in Figure~\ref{fig:regret_real}. Similar to the experiments on synthetic data, the vanishing (average) regret of AVP mean that the cumulative regret approaches an asymptote, whose value is dependent on the learning rate we choose. The benchmark algorithm does not have vanishing regret and thus shows linear cumulative regret.
 
\section{Conclusion and Future Directions}
As the online advertising industry continues to grow, advertisers are presented with a plethora of channels and platforms to promote their products. However, effectively allocating their budget across multiple platforms poses a significant challenge, as each platform operates with a different auction format that may not be incentive compatible. An advertiser's goal is, nevertheless, to achieve maximum utility within her budget constraints, which requires balancing current and future bidding opportunities across platforms with distinct sets of competitors. Furthermore, advertisers often lack prior knowledge of their competitors' bidding behaviors in each platform, adding further complexity to the problem.

We present an optimal budget management strategy in multi-platform advertising markets, where advertisers compete in auctions that may have varying and potentially non-incentive compatible formats. We propose a value-pacing strategy that allocates an advertiser's budget across multiple platforms and maximizes her expected (quasi-linear) utility. We also investigate the online setting, where the advertiser must submit bids while learning about other bidders' bids and their own bid distributions over time. Our adaptive budget optimization no-regret algorithm has $O(T^{3/4})$ regret under the full-information setting and empirically outperforms the current benchmark in synthetic and real datasets. Our work contributes to the literature by expanding upon previous research and presenting a general solution applicable to a broader range of auctions and platforms.

Our work suggests several avenues for future technical exploration. Firstly, for the online setting, it would be beneficial to identify weaker assumptions than those required for Theorem \ref{thm:online}, which necessitates thrice differentiability and strong convexity of the dual function. Although our empirical results suggest that Algorithm \ref{alg:online} produces desired outcomes even in the absence of these assumptions, a comprehensive empirical evaluation across various regimes would be valuable. Additionally, developing a better regret bound for the online setting in the full-information scenario and extending our algorithm to the bandit setting through careful analysis are promising directions for future research. Finally, evaluating our algorithm on larger, more complex datasets is essential for assessing its performance in a real-world scenario, where it can compete with actual bidders.

\bibliographystyle{ACM-Reference-Format}
\bibliography{ref}

\newpage
\onecolumn 
\appendix
\section{Appendix}
\label{sec:appendix}

\subsection{Proof of Theorem~\ref{thm:offline}}
\label{sec:offlineproof}

\offlinethm*

\begin{proof}
Recall that $q(\mu,\mathbf{V},{\Fcal}, \rho)$ is the dual of the optimization problem in Equation~\eqref{eq:offlineoptproblem}, as simplified in Equation~\eqref{eq:substituteddualmany}: 
\[q(\mu,\mathbf{V},{\Fcal}, \rho) = {\mu}\rho + \E\left[\sum_{j=1}^J {v_j}\expalloc\left(\sigma_j\left(\frac{v_j}{1+\mu}\right)\right)-(1+\mu)\exppay\left(\sigma_j\left(\frac{v_j}{1+\mu}\right)\right)\right].\] See that $q(\mu,\mathbf{V},{\Fcal}, \rho)$ is convex as a function of ${\mu}$ since the objective function of a dual of a maximization problem is always convex. Moreover, when ${\mu}>\frac{JU}{\rho}$ , it follows that
\begin{equation*}
q(\mu,\mathbf{V},{\Fcal}, \rho) \geq {\mu}\rho > JU \geq q(0,\mathbf{V},{\Fcal}, \rho),
\end{equation*}
where the first inequality holds because we assume that the expected utility of bidding any values less than or equal to the realized values is always non-negative, and the expected utility at the best response value must be at least greater than or equal to the expected utility at any $v'\leq v$, which is greater than or equal to $0$. The second inequality holds when $\mu > JU/\rho$, while the last inequality holds because the sum of the expected utility across all $J$ auctions is upper-bounded by the sum of values across auctions, which is itself upper-bounded by $JU$. As a result, the minimizer of $q(\mu,\mathbf{V},{\Fcal}, \rho)$ lies in the interval $[0,JU/\rho]$. Additionally, the dual is differentiable with respect to ${\mu}$, as $\expalloc$, $\exppay$ and $\sigma_j$ are continuously differentiable for each auction $j$.

Since $\mu^*$ is the minimizer of $q(\mu,\mathbf{V},{\Fcal}, \rho)$, we either have $\mu^* = 0$ or we can apply the first order condition to the function $q(\mu,\mathbf{V},{\Fcal}, \rho)$. We first note that the derivative of such function is as follows. 
\begin{equation}
\label{eq:derivdualqgeneral}
\begin{aligned}
\frac{\partial}{\partial {\mu}}q(\mu,\mathbf{V},{\Fcal}, \rho) &= \rho + \E\Bigg[\sum_{j=1}^J v_j\frac{d\expalloc}{db}\left(\sigma_j\left(\frac{v_j}{1+\mu}\right)\right)\sigma_j'\left(\frac{v_j}{1+\mu}\right)\left(\frac{-v_j}{(1+\mu)^2}\right)\\
&\qquad-\exppay\left(\sigma_j\left(\frac{v_j}{1+\mu}\right)\right)-(1+\mu)\frac{d\exppay}{db}\left(\sigma_j\left(\frac{v_j}{1+\mu}\right)\right)\sigma_j'\left(\frac{v_j}{1+\mu}\right)\left(\frac{-v_j}{(1+\mu)^2}\right)\Bigg]\\
&= \rho - \E\Bigg[\sum_{j=1}^J \exppay\left(\sigma_j\left(\frac{v_j}{1+\mu}\right)\right)\Bigg].
\end{aligned}
\end{equation}
The last line follows from the substitution of $v = \frac{v_j}{1+\mu}$ in Equation~\eqref{eq:sigmaderiv} where we have
\begin{equation}
\begin{aligned}
v\frac{d\expalloc}{db}(\sigma_j(v)) - \frac{d\exppay}{db}(\sigma_j(v)) &= 0\\
\left(\frac{v_j}{1+\mu}\right)\frac{d\expalloc}{db}\left(\sigma_j\left( \frac{v_j}{1+\mu}\right)\right) - \frac{d\exppay}{db}\left(\sigma_j\left( \frac{v_j}{1+\mu}\right)\right) &= 0.
\end{aligned}
\end{equation}
Since the function $q(\mu, \Fcal,\mathbf{V},\rho)$ is continuous with respect to ${\mu}$ and attains its minimum value in the closed and bounded interval $[0, \frac{JU}{\rho}]$, Weierstrass Theorem guarantees the existence of a minimizer in this interval. To prove the theorem, it suffices to demonstrate the modified KKT conditions outlined in Theorem 5.1.5 in \cite{bertsekas1997nonlinear}, which are:
\begin{enumerate}
    \item Primal feasibility, where the budget constraint is satisfied: $\sum_{j=1}^J\E\left[\exppay\left(\sigma_j\left(\frac{v_j}{1+\mu^*}\right)\right)\right]\leq\rho.$ This follows from the first order condition applied to the function $q(\mu,\mathbf{V},{\Fcal}, \rho)$ and the minimizer $\mu^*$, where we have 
    \begin{equation*}
    \frac{\partial q(\mu^*,\Fcal,\mathbf{V},\rho)}{\partial {\mu}}\geq 0 \quad\qquad\Leftrightarrow\quad\qquad\sum_{j=1}^J\E\left[\exppay\left(\sigma_j\left(\frac{v_j}{1+\mu^*}\right)\right)\right]\leq\rho.
    \end{equation*}
    The first inequality is also true when $q^* = 0$, since we optimize in the space $[0, \infty)$.
    \item Dual feasibility, $\mu^*\geq 0$,
    which is true since $\mu^*\in[0, JU/\rho]$.
    \item Complementary slackness, i.e.
    \begin{equation*}
    \mu^*\cdot\left\{\rho - \sum_{j=1}^J\E\left[\exppay\left(\sigma_j\left(\frac{v_j}{1+\mu^*}\right)\right)\right]\right\} = 0.
    \end{equation*}
    This holds because the first order condition  applied to the function $q(\mu,\mathbf{V},{\Fcal}, \rho)$ and the minimizer $\mu^*$ says that
    \begin{equation*}
    \mu^*\cdot\frac{\partial q(\mu^*,\Fcal,\mathbf{V},\rho)}{\partial {\mu}} = 0\qquad\Leftrightarrow\qquad \mu^*\cdot\left\{\rho - \sum_{j=1}^J\E\left[\exppay\left(\sigma_j\left(\frac{v_j}{1+\mu^*}\right)\right)\right]\right\} = 0.
    \end{equation*}
\end{enumerate}
\end{proof}

\subsection{Proof of Proposition~\ref{prop:offlinetoonline}}
\label{sec:prop1proof}
\offlinetoonline*
\begin{proof}
We have
\begin{equation}
\begin{aligned}
\E[OPT(\bm{X}, \rho)] 
&\overset{(1)}{\leq} \E\min_{\mu\geq 0}\Bigg[\mu\rho T + \underset{\substack{\left\{b_{j,t}(\cdot):\mathcal{V}\rightarrow\mathbb{R}_{\geq 0}\right\}\\t\in[T],j\in[J]}}{\max}\quad \sum_{t=1}^T\sum_{j=1}^J v_{j,t}\expalloc(b_{j,t}(v_{j,t})) - (1+\mu)\exppay(b_{j,t}(v_{j,t}))\Bigg]\\
&\leq \min_{\mu\geq 0} \E\Bigg[\mu\rho T+\underset{\substack{\left\{b_{j,t}(\cdot):\mathcal{V}\rightarrow\mathbb{R}_{\geq 0}\right\}\\t\in[T],j\in[J]}}{\max}\quad \sum_{t=1}^T\sum_{j=1}^J v_{j,t}\expalloc(b_{j,t}(v_{j,t})) - (1+\mu)\exppay(b_{j,t}(v_{j,t}))\Bigg]\\
&= \min_{\mu\geq 0} T \E\Bigg[\mu\rho + \underset{\substack{\left\{b_{j}(\cdot):\mathcal{V}\rightarrow\mathbb{R}_{\geq 0}\right\}\\j\in[J]}}{\max}\quad \sum_{j=1}^J v_{j}\expalloc(b_{j}(v_{j})) - (1+\mu)\exppay(b_{j}(v_{j}))\Bigg]\\
&\overset{(2)}{=} T\cdot Z(\Fcal, \mathbf{V}, \rho),
\end{aligned}
\end{equation}
where inequality~$(1)$ follows from weak duality, and equation~$(2)$ holds because of the strong duality implied by Theorem~\ref{thm:offline}.
\end{proof}

\subsection{Sufficient Conditions for Theorem~\ref{thm:offline} in GFP Auctions} 
\label{sec:sufficient}
We now look at one of the commonly used auctions, generalized first price (GFP) auction to investigate the properties that imply the expected allocation, expenditure and best response functions to be continuously differentiable almost everywhere. 

We consider a GFP auction in a position auction with $(n+1)$ bidders and $k$ positions. The discount factors are $\alpha_1 > \alpha_2 > \ldots > \alpha_k > 0$. In a GFP auction, the bidder with the $i^{\text{th}}$ highest bid obtains the $i^{\text{th}}$ position with CTR $\alpha_i$ and pays her bid, for $i = 1,\ldots,k$; if $i>k$, she gets and pays nothing. The bids of bidders $1,2,\ldots,n$ are assumed to be i.i.d., with their bids drawn from a continuous distribution $F$. The cumulative distribution function of $F$ is defined as $\mathbb{P}(\theta_i \leq b) = F(b)$ for all $i \in {1, 2, \ldots, n}$, where $\theta_1, \ldots, \theta_n$ are random variables representing the other bids. We assume that the distribution $F$ is non-negative.

The expected allocation function for the last bidder, when bidding $b$, is
\begin{equation*}
\begin{aligned}
    \bar{x}(b) &= \E_F[x_{n+1}(\bbf_{-(n+1)},b)] \\
    &= \alpha_1\binom{n}{n}[F(b)]^n + \alpha_2\binom{n}{n-1}[F(b)]^{n-1}[1-F(b)]^{1} + \ldots +\\ &\qquad\alpha_{k}\binom{n}{n-(n-1)}[F(b)]^{n-({k}-1)}[1-F(b)]^{{k}
    -1}\\
    &= \sum_{i=1}^{k}\alpha_i\binom{n}{n-i+1}[F(b)]^{n-i+1}[1-F(b)]^{i-1},
\end{aligned}
\end{equation*}
where the second equality follows from the fact that the probability that a bidder with bid $b$ gets the $i^{\text{th}}$ position equals to the probability that exactly $(n-(i-1))$ other bidders bid lower than $b$ multiplied by the probability that the rest $(i-1)$ other bidders bid higher than $b$ multiplied by the number of ways of choosing such $(n-i+1,i-1)$ bidder partitions, which is $\binom{n}{n-i+1}[F(b)]^{n-i+1}[1-F(b)]^{i-1}$. This function is continuously differentiable almost everywhere when $F$ is. The expected expenditure function for the last bidder is $\bar{p}(b) = \E_F[p_{n+1}(\bbf_{-(n+1)},b)] = \E_F[bx_{n+1}(\bbf_{-(n+1)},b)] = b\bar{x}(b)$, which is also continuously differentiable almost everywhere when $F$ is. It is possible to relax the assumption that all bids are i.i.d. When the bid of each bidder comes from a continuously differentiable distribution, both the expected allocation and expenditure functions are also continuously differentiable.

Furthermore, we claim the following lemma.
\begin{lemma}
In a GFP auction, when 
\begin{enumerate}
    \item the distribution of other bids $F$ is strictly monotone and continuously differentiable,
    \item the expected allocation function $\bar{x}:\mathbb{R}\rightarrow[0,1]$ is strictly increasing, continuously differentiable and concave,
\end{enumerate}
the best response function is unique and almost everywhere continuously differentiable.
\end{lemma}
\begin{proof}
We want to show that $\sigma_j(v)$ is the unique maximizer of $v\bar{x}(b)-\bar{p}(b)$ over $b\in\mathbb{R}$ and it is almost everywhere continuously differentiable. To do so, we first use the property of the GFP auctions and rewrite the maximization problem as
\begin{equation}
\label{eq:equivGFP}
\max_{b\in\mathbb{R}}v\bar{x}(b)-\bar{p}(b) = \max_{b\in\mathbb{R}}v\bar{x}(b)-\bar{x}(b)b = \max_{c\in[0,1]}vc - c\bar{x}^{-1}(c),
\end{equation}
where the last equality follows from substituting $c = \bar{x}(b)$ and $b = \bar{x}^{-1}(c)$. We then show that the function $vc - c\bar{x}^{-1}(c)$ is a strongly concave function of $c$ that has a unique maximum by proving that the function $g(c) = c\bar{x}^{-1}(c)$ is strongly convex in $[0,1]$. We also show that the function $\max_{c\in[0,1]}vc - c\bar{x}^{-1}(c)$ is equivalent to the convex conjugate of $g$ ($g^*(c)$), which later implies $\sigma_j(v) = \arg\max_{c\in[0,1]}vc - c\bar{x}^{-1}(c) = \arg\max_{c\in[0,1]}g^*(c)$ to be unique and continuously differentiable.

First, we prove that when $\bar{x}$ is strictly monotone, continuously differentiable, and concave, and its derivative is upper-bounded by $P$, the function $g(c) = c\bar{x}^{-1}(c)$ is $2/P$-strongly convex in $[0,1]$, where $\bar{x}^{-1}:[0,1]\rightarrow\mathbb{R}$ is the inverse function of $\bar{x}$ (the existence is implied since $\bar{x}$ is strictly monotone). See that for any $(c,y)$, we have
\begin{equation*}
    \bar{x}^{-1}(c)\overset{}{\geq} \bar{x}^{-1}(y) + \frac{d}{dy}\bar{x}^{-1}(y)(c-y) \overset{}{=} \bar{x}^{-1}(y) +\frac{1}{\frac{d\bar{x}(b)}{db}\big|_{b=y}}(c-y),
\end{equation*}
where the inequality follows from the concavity of $\bar{x}$, which implies the convexity of $\bar{x}^{-1}$, and the equality follows from the derivative of the inverse function. We will use this inequality to show that the function $g(c) = c\bar{x}^{-1}(c)$ is $2/P$- strongly convex, as follows.
\begin{equation*}
\begin{aligned}
    c\bar{x}^{-1}(c) &\geq c\bar{x}^{-1}(y) + \frac{c}{\frac{d\bar{x}(b)}{db}\big|_{b=y}}(c-y)\\
    \Leftrightarrow c\bar{x}^{-1}(c) &\geq (y + c-y)\bar{x}^{-1}(y) + \frac{y+c-y}{\frac{d\bar{x}(b)}{db}\big|_{b=y}}(c-y)\\
    \Leftrightarrow c\bar{x}^{-1}(c) &\geq y\bar{x}^{-1}(y) + \left(\bar{x}^{-1}(y) + \frac{y}{\frac{d\bar{x}(b)}{db}\big|_{b=y}}\right)(c-y) + \frac{(c-y)^2}{\frac{d\bar{x}(b)}{db}\big|_{b=y}}\\
    \Leftrightarrow c\bar{x}^{-1}(c)&\overset{(1)}{\geq} y\bar{x}^{-1}(y) + \left(\bar{x}^{-1}(y) + \frac{y}{\frac{d\bar{x}(b)}{db}\big|_{b=y}}\right)(c-y) + \frac{1}{P}(c-y)^2\\
    \Leftrightarrow g(c) &\geq g(y) + g'(c)(c-y) + \frac{1}{P}(c-y)^2\qquad\Leftrightarrow\qquad\text{$g(c)$ is $2/P$-strongly convex.}
\end{aligned}
\end{equation*}
Here, Inequality~(1) holds because the derivative of $\bar{x}$, $\left|\frac{d\bar{x}}{db}(y)\right|$, is upper-bounded by $P$.

Then, we show that $\sigma_j(v)$ is the unique maximizer of $v\bar{x}(b)-\bar{p}(b)$ over $b\in\mathbb{R}$ and $\sigma_j(v)$ is almost everywhere continuously differentiable. See that when the last bidder's value is $v$, $v\bar{x}(b)-\bar{p}(b)\leq0$ when $b= 0$\footnote{Note that we assume all bids are non-negative, and bidders with zero bid are not allocated since $F(0) = 0$.} or $b\geq v$, and $v\bar{x}(b)-\bar{p}(b)\leq v\leq U$, so $v\bar{x}(b)-\bar{p}(b)$ attains its maximum when $b$ is in the compact space $[0,v]$. Moreover, $\arg\max_{b\in\mathbb{R}}v\bar{x}(b)-\bar{p}(b) = \arg\max_{b\in\mathbb{R}}v\bar{x}(b)-\bar{x}(b)b = \arg\max_{c\in[0,1]}vc - c\bar{x}^{-1}(c)$, as shown in Equation~\eqref{eq:equivGFP}. Since $g(c) = c\bar{x}^{-1}(c)$ is strongly convex, the function $vc - c\bar{x}^{-1}(c)$ is strongly concave (as a function of $c$). It is also bounded above by $U$. Since a bounded strongly concave function has a unique maximum in a compact space, the function $vc - c\bar{x}^{-1}(c)$ attains its unique maximum for some $c^*\in[0,1]$.

Furthermore, if $g^*$ is the conjugate of $g$, we have $g^*(b) = \max_{c\in[0,1]}vc-g(c) = \max_{b\in\mathbb{R}}v\bar{x}(b)-\bar{p}(b) $. Since $g$ is $2/P$-strongly convex, $g^*$ is differentiable everywhere and $g^{*'}(v) = \arg\max_{c\in[0,1]}vc-g(c) = \sigma_j(v)$ (\cite{stromberg2009note}). Further, $g^{*'}$ is Lipschitz continuous as a convex conjugate of a strongly convex function, and a Lipschitz continuous function must be continuously differentiable almost everywhere, hence $\sigma_j(v)$ is continuously differentiable almost everywhere.
\end{proof}

\subsection{Proof of Theorem~\ref{thm:online}}
\label{sec:proof-online-thm}

\onlinethm*
\begin{proof}
There are 7 steps of the proof, as outlined in Section~\ref{sec:online}. Suppose that, for all $t\in[T]$ and $j\in[J]$, there exists $\delta_{j,t}>0$ such that $\sup_{v\in\mathbb{R}_{\geq 0}}|\sigma_j(x) - \hat{\sigma}_{j,t}(x)|\leq\delta_{j,t}$. We specify the conditions for such $\{\delta_{j,t}\}_{j\in[J],t\in[T]}$ in Lemma \ref{lemma:BR-estimate}.

\textbf{Step 1: upper-bound the performance in hindsight.} Looking at the in-hindsight benchmark in Equation~\eqref{eq:generalonlinebenchmark} and its dual, by applying weak duality,
as implied by Proposition~\ref{prop:offlinetoonline}, we have
\begin{equation*}
OPT(\bm{X},\rho) \leq T\cdot Z(\Fcal, \mathbf{V},\rho) \leq T\cdot\inf_{\mu\geq 0}{\Gamma}(\mu),
\end{equation*}
where
\begin{equation*}
\begin{aligned}
{\Gamma}(\mu) &= \mu\rho + \sum_{j=1}^J\E_{V_j}\left[v_j\expalloc\left(\sigma_j\left(\frac{v_j}{1+\mu}\right)\right) - (1+\mu)\exppay\left(\sigma_j\left(\frac{v_j}{1+\mu}\right)\right)\right].
\end{aligned}
\end{equation*}

The following characteristics of the dual function ${\Gamma}(\mu)$ will be used later:
\begin{enumerate}
\item From the analysis for the static case, we know that the minimizer of $q(\mu,\mathbf{V},\Fcal,\rho)$  lies in $[0,JU/\rho]$, which is a compact space. Hence, $\mu^*$, the minimizer of ${\Gamma}(\mu)$ is also attained in $[0,JU/\rho]$ and the minimum value is ${\Gamma}(\mu^*)$.
\item For a fixed $\mu$, let $U(\mu)$ and $G(\mu)$ be the expected utility and expenditure when the optimal bids are used (based on the offline results) and the multiplier is $\mu$, i.e., $U(\mu) = \sum_{j=1}^J\mathbb{E}_{V_j}\left[v_j\expalloc\left(\sigma_j\left(\frac{v_j}{1+\mu}\right)\right) - \exppay\left(\sigma_j\left(\frac{v_j}{1+\mu}\right)\right)\right]$ and $G(\mu) = \sum_{j=1}^J\mathbb{E}_{V_j}\left[\exppay\left(\sigma_j\left(\frac{v_j}{1+\mu}\right)\right)\right]$. We have ${\Gamma}(\mu) = U(\mu) +\mu(\rho - G(\mu))$.
\item Since $\Gamma(\mu)$ is thrice differentiable, by dominated convergence theorem (DCT), the derivative of the expectation is equal to the expectation of the derivative, hence we get the following derivative: \begin{equation*}
{\Gamma}'(\mu) = \rho - \sum_{j=1}^J\mathbb{E}_{V_j}\left[\exppay\left(\sigma_j\left(\frac{v_j}{1+\mu}\right)\right)\right].
\end{equation*}
Moreover, we have ${\Gamma}'(\mu) = \rho - G(\mu).$
\item Since the dual function $\Gamma(\mu)$ is thrice differentiable, the KKT condition implies that we have $\mu^*{\Gamma}'(\mu^*) = \mu^*(\rho - G(\mu^*)) = 0$, $\mu^*\geq 0$, and ${\Gamma}'(\mu^*) = \rho - G(\mu^*)\geq 0$ when $\mu^*\in\arg\min_{\mu\geq 0}\Gamma(\mu) = \arg\min_{\mu\geq 0}{\Gamma}(\mu)$.
\end{enumerate}

\textbf{Step 2: bound the difference in the expected allocation and expenditure as functions of $\delta_{j,t}$.} Let $\hat{G}_t$ be the expenditure at round $t$ when the estimated best response functions are used, i.e.
\begin{equation*}
\hat{G}_t(\mu) = \sum_{j=1}^J\mathbb{E}_{V_j}\left[\exppay\left(\hatsigma_{j,t}\left(\frac{v_j}{1+\mu}\right)\right)\right].
\end{equation*}
Since $|\sigma_j(x) - \hatsigma_{j,t}(x)|\leq\delta_{j,t}$, we have
\begin{equation}
\label{eq:expenditurebound}
\begin{aligned}
\Big|\hat{G}_t(\mu) - G(\mu)\Big| &\leq \sum_{j=1}^J\Bigg|\mathbb{E}_{V_j}\left[\exppay\left(\hatsigma_{j,t}\left(\frac{v_j}{1+\mu}\right)\right) - \exppay\left(\sigma_j\left(\frac{v_j}{1+\mu}\right)\right)\right]\Bigg|\\
&\leq L\sum_{j=1}^J\delta_{j,t} = \gamma_t.
\end{aligned}
\end{equation}
The first inequality follows from triangle inequality, while the second inequality holds because $\exppay$ is $L$-Lipschitz.

Let $\hat{U}_t$ be the expected utility at round $t$ when the estimated best responses are used, i.e.
\begin{equation*}
\hat{U}_t(\mu) = \sum_{j=1}^J\mathbb{E}_{V_j}\left[v_{j,t}\expalloc\left(\hatsigma_{j,t}\left(\frac{v_j}{1+\mu}\right)\right) - \exppay\left(\hatsigma_{j,t}\left(\frac{v_j}{1+\mu}\right)\right)\right],
\end{equation*}
and $U$ be the expected utility when the actual best responses are known, where the expectations here are taken with respect to the distribution of other bids and the distribution of the last bidder's values. Then, we have
\begin{equation}
\label{eq:utilitybound}
\begin{aligned}
\Big|\hat{U}_t(\mu) - U(\mu)\Big| &\leq \sum_{j=1}^J\Bigg|\mathbb{E}_{V_j}\Bigg[v_{j,t}\expalloc\left(\hatsigma_{j,t}\left(\frac{v_{j,t}}{1+\mu}\right)\right) - \exppay\left(\hatsigma_{j,t}\left(\frac{v_{j,t}}{1+\mu}\right)\right) \\
&\qquad\quad- \left(v_{j,t}\expalloc\left(\sigma_j\left(\frac{v_{j,t}}{1+\mu}\right)\right) - \exppay\left(\sigma_j\left(\frac{v_{j,t}}{1+\mu}\right)\right)\right)\Bigg]\Bigg|\\
&\leq \sum_{j=1}^J (UL+L)\delta_{j,t} := \Delta_t,
\end{aligned}
\end{equation}
where the first inequality follows from triangle inequality, and the second inequality holds because we assume that both $\bar{x}_j$ and $\bar{p}_j$ are $L$-Lipschitz.

\textbf{Step 3: lower-bounding the regret of policy AVP.} We consider an alternate framework where we are allowed to bid after budget depletion. From the algorithm, we have $\tilde{B}_t = \rho T - \sum_{s=1}^{t-1}\sum_{j\in[J]}\hat{p}_{j,t}$ as the remaining budget at the beginning of period $t$ in the alternate framework. Let $\tau = \inf\{t\geq 1:\tilde{B}_{t+1}\leq JU\}$ be the last period where the remaining budget is larger than $JU$. Since $v_j/(1+\mu)\leq U$ for all auctions $j$ and $\mu\geq 0$, for any period $t\leq\tau$, the bid in this framework (and the original framework) is $\hatsigma_{j,t}\left(\frac{v_{j,t}}{1+\mu_t}\right)$ for auction $j$. The performance of the original and the alternate frameworks coincide up to time $\tau$, and the total utility of AVP is lower-bounded by
\begin{equation}
\label{eq:loweboundJU}
    \sum_{j=1}^J\E_{V_j}\left[\sum_{t=1}^{\tau\wedge T}v_{j,t}\hat{x}_{j,t} -\hat{p}_{j,t}\right]\geq \sum_{j=1}^J\E_{V_j}\left[\sum_{t=1}^Tv_{j,t}\hat{x}_{j,t} -\hat{p}_{j,t}\right] - JU\mathbb{E}\left[(T-\tau)^+\right],
\end{equation}
where the first term is a result of discarding the total utility after the budget is depleted and the inequality holds because $0\leq v_{j,t}\leq U$.

\textbf{Step 4: lower-bounding the budget depletion period $\tau$.} See that we have $\tau = \inf\{t\geq 1: \tilde{B}_{t+1}\leq JU\}$ and in the adaptive algorithm, $\mu_{t+1} = \mu_t +\epsilon(\sum_{j\in[J]}\hat{p}_{j,t}-\rho) - P_t$, where $P_t = \mu_t +\epsilon(\realpay -\rho) - P_{[0,\bar{\mu}]}\left(\mu_t +\epsilon(\realpay -\rho)\right)$ is the projection error. Reordering terms and summing up to period $\tau$, we have
\begin{equation*}
    \sum_{t=1}^\tau \left(\realpay-\rho\right) = \sum_{t=1}^\tau \frac{1}{\epsilon}(\mu_{t+1}-\mu_t) + \sum_{t=1}^\tau \frac{P_t}{\epsilon}.
\end{equation*}

For the left hand side, we have
\begin{equation*}
    \sum_{t=1}^\tau \left(\realpay-\rho\right) = \rho T - \tilde{B}_{\tau+1} - \rho\tau\geq\rho(T-\tau)-JU,
\end{equation*}
where the equality holds since $\tilde{B}_{\tau+1} = \rho T - \sum_{t=1}^\tau \realpay$ and the inequality holds because $\tilde{B}_{\tau+1}\leq JU$.

For the first term in the right hand side, we have
\begin{equation*}
\sum_{t=1}^\tau \frac{1}{\epsilon}(\mu_{t+1}-\mu_t) = \frac{\mu_{\tau+1}}{\epsilon} - \frac{\mu_1}{\epsilon} \leq \frac{\bar{\mu}}{\epsilon},
\end{equation*}
where the inequality holds because $\mu_t\in[0,\bar{\mu}]$.

For the projection error, we have
\begin{equation*}
\begin{aligned}
P_t &\leq P_t^+ = \left(\mu_t +\epsilon\left(\realpay-\rho\right) - P_{[0,\bar{\mu}]}\left(\mu_t +\epsilon\left(\realpay-\rho\right)\right)\right)^+\\
&= \left(\mu_t +\epsilon\left(\realpay-\rho\right)-\bar{\mu}\right)\mathbf{1}\left\{\mu_t + \epsilon\left(\realpay-\rho\right)>\bar{\mu}\right\}\\
&\leq \epsilon JU \mathbf{1}\left\{\mu_t + \epsilon\left(\realpay-\rho\right)>\bar{\mu}\right\},
\end{aligned}
\end{equation*}
where the second line holds since the projection error is only positive when $\mu_t + \epsilon(\realpay - \rho)>\bar{\mu}$ and the last line holds since $\mu_t\leq\bar{\mu}$ and $\realpay\leq JU$. We will then prove that since $\epsilon < 1/JU$, $P_t^+$ is always $0$. Consider the function $f(\mu) = \mu + \epsilon JU/(1+\mu)$. The function is non-decreasing when $\epsilon JU\leq 1.$ Then, we have
\begin{equation*}
\begin{aligned}
\mu_t + \epsilon\left(\realpay-\rho\right) &\leq \mu_t + \epsilon\left(\frac{JU}{1+\mu_t} - \rho\right) = f(\mu_t) -\epsilon\rho\\
&\leq f(\bar{\mu}) - \frac{\epsilon J U}{1+\bar{\mu}} = \bar{\mu},
\end{aligned}
\end{equation*}
where the first line holds because $\hat{p}_{j,t}\left(\frac{v_{j,t}}{1+\mu_t}\right)\leq \frac{v_{j,t}}{1+\mu_t}\leq \frac{U}{1+\mu_t}$, and the second line holds because f is non-decreasing, $\mu_t\leq\bar{\mu}$, and $\bar{\mu}\geq JU/\rho - 1$. So, there is no positive projection error when $\epsilon < 1/JU$. Putting things together, we have
\begin{equation}
\label{eq:timelowerbound}
T - \tau\leq \frac{\bar{\mu}}{\epsilon\rho} + \frac{JU}{\rho}.
\end{equation}

\textbf{Step 5: lower-bounding the utility per auction.} We will show that the utility per auction is close to ${\Gamma}(\mu^*)$. For a fixed $\mu_t$, the expected utility that the last bidder gets at round $t$ is
\begin{equation*}
\begin{aligned}
\hat{U}_t(\mu_t) &= \sum_{j=1}^J\E_{V_j}\left[v_{j,t}\expalloc\left(\hatsigma_{j,t}\left(\frac{v_{j,t}}{1+{\mu_t}}\right)\right) - \exppay\left(\hatsigma_{j,t}\left(\frac{v_{j,t}}{1+{\mu_t}}\right)\right)\right] \\
&\geq U(\mu_t) - \Delta_t = {\Gamma}(\mu_t) + \mu_t(G(\mu_t) - \rho) - \Delta_t,
\end{aligned}
\end{equation*}
where the inequality follows from step 2 and we recall that ${\Gamma}(\mu) = U(\mu) +\mu(\rho - G(\mu))$.

We want to do a second order expansion on $U(\mu)$ around $\mu^*$. First, we note that $U(\mu^*) = {\Gamma}(\mu^*)$ since $\mu^*(\rho - G(\mu^*))= 0$ from the last property on step 1. Then, the first order derivative of $U$ is $U'(\mu) = \Gamma'(\mu) - (\rho - G(\mu)) - \mu(-G'(\mu)) = \mu G'(\mu)$ since $\Gamma'(\mu) = \rho - G(\mu)$ from step 1. Furthermore, the second order derivative of $U$ is $U''(\mu) = \mu G''(\mu) + G'(\mu)$, where there exist $\bar{G}'>0$ and $\bar{G}''>0$ such that $|G'(\mu)|\leq \bar{G}'$ and $|G''(\mu)|\leq \bar{G}''$ since the dual function is thrice differentiable with bounded derivatives.

By taking a second order Taylor's expansion on $U(\mu_t)$ around $\mu^*$, for some $\xi$ in between $\mu_t$ and $\mu^*$, we have 
\begin{equation*}
    U(\mu_t) = U(\mu^*) + U'(\mu^*)(\mu_t-\mu^*) + \frac{U''(\mu)}{2}(\mu_t-\mu^*)^2,
\end{equation*}
and by substituting $U(\mu^*), U'(\mu^*)$ and $U''(\mu^*)$ and taking the expectation with respect to all the other bid distributions and the distributions of values, we get
\begin{equation*}
\begin{aligned}
U(\mu_t) &= {\Gamma}(\mu^*) + \mu^*G'(\mu^*)\mathbb{E}[\mu_t - \mu^*] + \mathbb{E}\left[\frac{\mu G''(\xi)+G'(\xi)}{2}(\mu_t-\mu^*)^2\right]\\
&\geq {\Gamma}(\mu^*) - \bar{G}'\mu^*|\mathbb{E}[\mu_t - \mu^*]| - \frac{\bar{\mu}\bar{G}''+\bar{G}'}{2}\mathbb{E}[(\mu_t-\mu^*)^2].
\end{aligned}
\end{equation*}

Thus, the total expected utility at round $t$ is lower-bounded by
\begin{equation}
\label{eq:lowerboundutilityroundt}
\hat{U}_t(\mu_t) \geq {\Gamma}(\mu^*) - \bar{G}'\mu^*|\mathbb{E}[\mu_t - \mu^*]| - \frac{\bar{\mu}\bar{G}''+\bar{G}'}{2}\mathbb{E}[(\mu_t-\mu^*)^2] - \Delta_t.
\end{equation}

\textbf{Step 6: upper-bounding the total expected errors.} We first upper-bound $\alpha_t = \mathbb{E}[(\mu_t-\mu^*)^2]$. By definition, we have
\begin{equation*}
\begin{aligned}
(\mu_{t+1}-\mu^*)^2 &= \left(P_{[0,\bar{\mu}]}\left(\mu_t-\epsilon\left(\rho-\realpay\right)\right)-\mu^*\right)^2\\
&\leq\left(\mu_t-\mu^* - \epsilon\left(\rho-\realpay\right)\right)^2 \\
&= (\mu_t-\mu)^2 - 2\epsilon(\mu_t-\mu^*)\left(\rho-\realpay\right) + \epsilon^2\left(\rho - \realpay\right)^2,
\end{aligned}
\end{equation*}
where the first inequality holds because $\mu^*\in[0,\bar{\mu}]$ when $\bar{\mu} = JU/\rho$. Dividing by $\epsilon$ and taking expectation w.r.t. all the other bid distributions and the distributions of values, we have
\begin{equation}
\label{eq:deltaeqhat}
\begin{aligned}
\frac{\mathbb{E}[(\mu_{t+1}-\mu^*)^2]}{\epsilon}&\leq\frac{\mathbb{E}[(\mu_t-\mu^*)^2]}{\epsilon} - 2\mathbb{E}[(\mu_t-\mu^*)(\rho-\realpay)] + \mathbb{E}[(\rho - \realpay)^2]\epsilon.
\end{aligned}
\end{equation}
For the second term, if $\phi_{j,t}$ denotes the payment that the last bidder gets in auction $j$ for bidding optimally on round $t$, i.e.
\begin{equation*}
    \phi_{j,t} = p_{j,n+1}\left(\bbf_{j,-(n+1)}, \sigma_j\left(\frac{v_{j,t}}{1+\mu_t}\right)\right),
\end{equation*}
we have
\begin{equation*}
\begin{aligned}
\mathbb{E}[(\mu_t-\mu^*)(\rho - \realpay)] &\overset{(1)}{=} \mathbb{E}[\mathbb{E}[(\mu_t-\mu^*)(\rho - \sum_{j\in[J]}\phi_{j,t} + \sum_{j\in[J]}\phi_{j,t} - \realpay)|\mu_t]] \\
&= \mathbb{E}[(\mu_t-\mu^*)(\rho - G(\mu_t) + G(\mu_t) - \hat{G}_t(\mu_t))]\\
&= \E[(\mu_t-\mu^*)(\rho - G(\mu^*) + G(\mu^*) - G(\mu_t))] + \E[(\mu_t-\mu^*)(G(\mu_t) - \hat{G}_t(\mu_t))]\\
&\overset{}{=} \mathbb{E}[(\mu_t-\mu^*)(G(\mu^*) - G(\mu_t))] + \mathbb{E}[\mu_t(\rho - G(\mu^*))] \\
&\qquad\qquad- \mathbb{E}[\mu^*(\rho - G(\mu^*))] + \E[(\mu_t-\mu^*)(G(\mu_t) - \hat{G}_t(\mu_t))]\\
&\overset{(2)}{\geq} \mathbb{E}[(\mu_t-\mu^*)(G(\mu^*) - G(\mu_t))] + \E[(\mu_t-\mu^*)(G(\mu_t) - \hat{G}_t(\mu_t))]\\
&\overset{(3)}{\geq} \mathbb{E}[\lambda(\mu_t-\mu^*)^2] + \mathbb{E}[(\mu_t-\mu^*)(G(\mu_t) - \hat{G}_t(\mu_t))]\\
&\overset{(4)}{\geq} \mathbb{E}[\lambda(\mu_t-\mu^*)^2] - \bar{\mu}\gamma_t.
\end{aligned}
\end{equation*}
Equation~$(1)$ holds because of the independence of $(V_1,\ldots,V_J)$ from $\mu_t$. Inequality~$(2)$ follows from $\mu_t\geq 0$, $\rho - G(\mu^*)\geq 0$, and $\mu^*(\rho-G(\mu^*)) = 0$. Inequality~$(3)$ holds because the dual function is $\lambda$ strongly convex and $G(\mu) = \rho - \Gamma(\mu)$, so we have
\begin{equation*}
\begin{aligned}
\lambda(\mu - \mu')^2 &\leq (\Gamma(\mu) - \Gamma(\mu'))(\mu-\mu')\\
&= (\rho - \Gamma(\mu') - \rho + \Gamma(\mu))(\mu-\mu')\\
&= (G(\mu') - G(\mu))(\mu-\mu').
\end{aligned}
\end{equation*}
Here, Inequality~$(4)$ follows from the expenditure bound in Inequality~\eqref{eq:expenditurebound} where $|\hat{G}_t(\mu) - G(\mu)|\leq \gamma_t$.

For the third term in Equation~\eqref{eq:deltaeqhat}, we have $\mathbb{E}[(\rho - \realpay)^2]\leq (JU)^2$ since $\rho\leq JU$, $\realpay\geq 0$, and the total payment $\realpay$ is at most the sum of the bids in all auctions, which is at most the sum of all values in all auctions, which is upper-bounded by $JU$. Substituting, we have
\begin{equation*}
    \frac{\alpha_{t+1}}{\epsilon} \leq \frac{\alpha_t}{\epsilon} - 2\lambda\epsilon\frac{\alpha_t}{\epsilon} + \bar{\mu}\gamma_t + \epsilon(JU)^2,
\end{equation*}
substituting recursively, we have
\begin{equation*}
\begin{aligned}
    \alpha_t &\leq \alpha_1(1-2\lambda\epsilon)^{t-1} + \frac{\epsilon J^2U^2}{2\lambda} + \bar{\mu}\sum_{k=1}^t(1-2\lambda\epsilon)^{t-k}\gamma_k\\
    &\leq \bar{\mu}^2(1-2\lambda\epsilon)^{t-1} + \frac{\epsilon J^2 U^2}{2\lambda} + \bar{\mu}\sum_{k=1}^t(1-2\lambda\epsilon)^{t-k}\gamma_k,
\end{aligned}
\end{equation*}
which further implies
\begin{equation}
\label{eq:upperbounddelta}
    \sum_{t=1}^T\alpha_t\leq\frac{J^2U^2\epsilon T}{2\lambda} + \bar{\mu}^2\cdot\frac{1}{2\lambda\epsilon} + \bar{\mu}\cdot\frac{1}{2\lambda\epsilon}\sum_{t=1}^T\gamma_t.
\end{equation}

We now upper-bound the first term  in Equation~\eqref{eq:deltaeqhat}. Let $r_t := \mu^*|\mathbb{E}[\mu_t-\mu^*]|$. There are two cases. 

Case 1: when $\mu^*\leq\epsilon\rho$. Since $0\leq\mu_t,\mu^*\leq\bar{\mu}$, we have $r_t = \mu^*|\mathbb{E}[\mu_t-\mu^*]|\leq \epsilon\rho\bar{\mu}$, implying $\sum_{t=1}^Tr_t\leq T\epsilon\rho\bar{\mu}$. 

Case 2: when $\mu^*>\epsilon\rho$. From the update rule of the dual variables, we have $\mu_{t+1} = \mu_t +\epsilon(\sum_{j\in[J]}\hat{p}_{j,t}-\rho) - P_t$, where $P_t$ is the projection error. Subtracting $\mu^*$, multiplying by $\mu^*$, taking expectations, then the absolute value, we obtain that 
\begin{equation}
\label{eq:rt}
\begin{aligned}
r_{t+1} & = \mu^*\big|\mathbb{E}[\mu_{t+1}-\mu^*]\big| \overset{}{=}\big|\mu^*(\E[\mu_t - \mu^* + \epsilon(\sum_{j\in[J]}\hat{p}_{j,t}-\rho) - P_t])\big|\\
&\overset{}{=}\big|\mu^*\E[\mu_t - \mu^*] + \epsilon\mu^*(\E[\hat{G}_t(\mu_t)]-\rho) - \mu^*\E[P_t])\big|\\
&\overset{(1)}{=} \big|\mu^*\mathbb{E}[\mu_t-\mu^*] + \epsilon\mu^*(\mathbb{E}[\hat{G}_t(\mu_t) - G(\mu^*)]) - \mu^*\mathbb{E}[P_t]\big|\\
&= \big|\mu^*\mathbb{E}[\mu_t-\mu^*] + \epsilon\mu^*(\mathbb{E}[\hat{G}_t(\mu_t) - G(\mu_t) + G(\mu_t) - G(\mu^*)]) - \mu^*\mathbb{E}[P_t]\big|\\
&\overset{(2)}{\leq} \big|\mu^*\mathbb{E}[\mu_t-\mu^*] + \epsilon\mu^*(\mathbb{E}[G(\mu_t) - G(\mu^*)]) - \mu^*\mathbb{E}[P_t]\big| + \epsilon\bar{\mu}\big|\mathbb{E}[\hat{G}_t(\mu_t) - G(\mu_t)]\big|\\
&\overset{(3)}{=} \bigg|(1+\epsilon G'(\mu^*))\mu^*\mathbb{E}[\mu_t-\mu^*] + \frac{G''(\xi)}{2}\epsilon\mu^*\mathbb{E}[(\mu_t-\mu^*)^2]-\mu^*\mathbb{E}[P_t]\bigg| + \epsilon\bar{\mu}\big|\mathbb{E}[\hat{G}_t(\mu_t) - G(\mu_t)]\big|\\
&\overset{(4)}{\leq}\bigg|(1+\epsilon G'(\mu^*))\mu^*\mathbb{E}[\mu_t-\mu^*] + \frac{G''(\xi)}{2}\epsilon\mu^*\mathbb{E}[(\mu_t-\mu^*)^2]-\mu^*\mathbb{E}[P_t]\bigg|+ \epsilon\bar{\mu}\gamma_t,
\end{aligned}
\end{equation}
where Equation~$(1)$ holds because $\mu^*(G(\mu^*)-\rho) = 0$ while Inequality~$(2)$ follows from triangle inequality and $\mu^*\leq\bar{\mu}$. Equation~($3)$ follows from the Taylor's expansion of $G(\mu_t)$ around $\mu^*$; since $G$ is twice differentiable, for some $\xi$ between $\mu_t$ and $\mu^*$, we have 
\begin{equation*}
    G(\mu_t) = G(\mu^*) + (\mu_t-\mu^*)G'(\mu^*) + \frac{G''(\xi)}{2}(\mu_t-\mu^*)^2.
\end{equation*}
Lastly, Inequality~$(4)$ follows from Inequality \eqref{eq:expenditurebound} where $|\hat{G}_t(\mu) - G(\mu)|\leq \gamma_t$.

We now bound the projection error $P_t = \mu_t +\epsilon(\realpay -\rho) - P_{[0,\bar{\mu}]}\left(\mu_t +\epsilon(\realpay -\rho)\right)$. The projection error satisfies
\begin{equation*}
    P_t\overset{(a)}{\geq} (\mu_t+\epsilon(\realpay-\rho))\mathbf{1}\{\mu_t+\epsilon(\realpay-\rho)<0\} \overset{(b)}{\geq}-\epsilon\rho\mathbf{1}\{\mu_t<\epsilon\rho\},
\end{equation*}
where (a) holds because 1) when $\mu_t+\epsilon(\realpay-\rho)>\bar{\mu}$, $P_t = \mu_t+\epsilon(\realpay-\rho) - \bar{\mu}\geq 0$, 2) when $0\leq\mu_t+\epsilon(\realpay-\rho)\leq\bar{\mu}$, $P_t = \mu_t+\epsilon(\realpay-\rho)\geq 0$, and 3) when $\mu_t+\epsilon(\realpay-\rho)<0$, $P_t = \mu_t+\epsilon(\realpay-\rho)$. Furthermore, (b) holds because $\mu_t\geq 0$ and the expenditure is non-negative. Taking expectation then applying Markov's inequality, we have
\begin{equation*}
    -\mathbb{E}[P_t]\leq\epsilon\rho\mathbb{P}(\mu_t<\epsilon\rho)\leq\epsilon\rho\mathbb{P}\left(|\mu_t-\mu^*|\geq|\mu^*-\epsilon\rho|\right)\leq\frac{\epsilon\rho\alpha_t}{(\mu^*-\epsilon\rho)^2}\,,
\end{equation*}
where $\alpha_t = \mathbb{E}[(\mu_t-\mu^*)^2]$. The second inequality holds because in this case, $\mu^* > \epsilon\rho$, hence when $\mu_t<\epsilon\rho$, we have $\mu_t < \epsilon\rho < \mu^*$ and $\mu^* - \epsilon\rho < \mu^* - \mu_t$, which implies the second inequality.

Moreover, since the dual function is $\lambda$ strongly convex by our assumption, the second derivative of $\Gamma$, $-G'(\mu)$, is at least $\lambda$, so we have $\sup_{\mu\in[0,\bar{\mu}]}G'(\mu)\leq -\lambda$. Let $\bar{G}''$ be the upper-bound $\sup_{\mu\in[0,\bar{\mu}]}G''(\mu)$, which exists since the dual function is thrice differentiable with bounded derivatives. Applying the triangle inequality in Equation~\eqref{eq:rt}, we have
\begin{equation}
\begin{aligned}
\label{eq:prevprevbound}
    r_{t+1}&\leq |(1+\epsilon G'(\mu^*))\mu^*\mathbb{E}[\mu_t-\mu^*]| + \bigg|\frac{G''(\xi)}{2}\epsilon\mu^*\mathbb{E}[(\mu_t-\mu^*)^2]\bigg|+|\mu^*\mathbb{E}[P_t]|+ \epsilon\bar{\mu}\gamma_t\\
    &\leq(1-\epsilon\lambda)r_t + \frac{\bar{G}''}{2}\epsilon\mu^*\alpha_t + \frac{\bar{\mu}\epsilon\rho\alpha_t}{(\mu^*-\epsilon\rho)^2} + \epsilon\bar{\mu}\gamma_t = (1-\epsilon\lambda)r_t + \epsilon\alpha_t\left(\frac{\bar{\mu}\rho}{(\mu^*-\epsilon\rho)^2} + \frac{\bar{\mu}\bar{G}''}{2}\right) + \epsilon\bar{\mu}\gamma_t,
\end{aligned}
\end{equation}
where we substitute $\sup_{\mu\in[0, \bar{\mu}]}G'(\mu)\leq - \lambda$, $\alpha_t = \E[(\mu_t - \mu^*)^2]$, and the upper-bound on the absolute value of projection $|\E[P_t]|$. Let $R = \frac{\bar{\mu}\rho}{(\mu^*-\epsilon\rho)^2} + \frac{\bar{\mu}\bar{G}''}{2}$, applying this recursively, we obtain
\begin{equation*}
r_t\leq\bar{\mu}^2(1-\lambda\epsilon)^{t-1} + \epsilon R\sum_{s=1}^{t-1}(1-\lambda\epsilon)^{t-1-s}\alpha_s + \epsilon\bar{\mu}\sum_{k=1}^t(1-2\lambda\epsilon)^{t-k}\gamma_k,
\end{equation*}
and
\begin{equation}
\label{eq:upperboundrt}
\sum_{t=1}^T r_t\leq\frac{\bar{\mu}^2}{\lambda\epsilon} + \frac{R}{\lambda}\sum_{s=1}^T\alpha_s+\frac{\bar{\mu}}{2\lambda}\sum_{t=1}^T\gamma_t.
\end{equation}

\textbf{Step 7: putting everything together.} We have an upper-bound on the optimal in-hindsight solution from step 1, i.e., we have $OPT(\bm{X},\rho) \leq T\cdot\inf_{\mu\geq 0}{\Gamma}(\mu)$. We can also get a lower-bound on the total expected utility of the AVP algorithm by combining Equation~\eqref{eq:loweboundJU} from step 3 and Equation~\ref{eq:lowerboundutilityroundt} from step 5. We now have the following:
\begin{equation*}
\begin{aligned}
\text{Reg}^{\text{AVP}}(T,\rho) &\leq T\cdot\Gamma(\mu^*) - \sum_{t=1}^T\left({\Gamma}(\mu^*) - \bar{G}'\mu^*|\mathbb{E}[\mu_t-\mu^*]| - \left(\frac{\bar{\mu}\bar{G}''+\bar{G}'}{2}\right)\mathbb{E}[(\mu_t-\mu^*)^2] - \Delta_t\right) \\
&\qquad\qquad+ JU\mathbb{E}[(T-\tau)^+]\\
&\leq \bar{G}'\sum_{t=1}^T r_t + \left(\frac{\bar{\mu}\bar{G}''+\bar{G}'}{2}\right)\sum_{t=1}^T\alpha_t + \sum_{t=1}^T\Delta_t + JU\left(\frac{\bar{\mu}}{\epsilon\rho} + \frac{JU}{\rho}\right),
\end{aligned}
\end{equation*}
where the last line follows from Equation~\eqref{eq:timelowerbound} and substituting $r_t$ and $\alpha_t$ from step 6. There are two cases.

Case 1: when $\mu^*\leq\epsilon\rho$. By substituting the upper-bound on the sum of $r_t$, where we have $\sum_{t=1}^T r_t\leq T\epsilon\rho\bar{\mu}$, and substituting the upper-bound on $\sum_{t=1}^T\alpha_t$ in Equation~\eqref{eq:upperbounddelta} from step 6, we have
\begin{equation*}
\begin{aligned}
RHS &\leq \bar{G}'T\epsilon\rho\bar{\mu} + \left(\frac{\bar{\mu}\bar{G}''+\bar{G}'}{2}\right)\cdot\left(\frac{m^2U^2\epsilon T}{2\lambda} + \frac{\bar{\mu}^2}{2\lambda\epsilon}+\frac{\bar{\mu}}{2\lambda\epsilon}\sum_{t=1}^T\gamma_t\right) \\
&\quad  + \sum_{t=1}^T\Delta_t + JU\left(\frac{\bar{\mu}}{\epsilon\rho} + \frac{JU}{\rho}\right).
\end{aligned}
\end{equation*}

Case 2: when $\mu^*>\epsilon\rho$. We have
\begin{equation*}
\begin{aligned}
RHS &= \bar{G}'\sum_{t=1}^T r_t + \left(\frac{\bar{\mu}\bar{G}''+\bar{G}'}{2}\right)\sum_{t=1}^T\alpha_t + \sum_{t=1}^T\Delta_t + JU\left(\frac{\bar{\mu}}{\epsilon\rho} + \frac{JU}{\rho}\right)\\
&\leq \frac{\bar{G}'\bar{\mu}^2}{\lambda\epsilon} + \frac{\bar{G}'}{\lambda}\left(\frac{\bar{\mu}\rho}{(\mu^*-\epsilon\rho)^2)} + \frac{\bar{\mu}\bar{G}''}{2}\right)\sum_{s=1}^T\alpha_s+ \frac{\bar{G}'\bar{\mu}}{2\lambda}\sum_{t=1}^T\gamma_t+\left(\frac{\bar{\mu}\bar{G}''+\bar{G}'}{2}\right)\sum_{t=1}^T\alpha_t\\
&\quad+ \sum_{t=1}^T\Delta_t+ JU\left(\frac{\bar{\mu}}{\epsilon\rho} + \frac{JU}{\rho}\right)\\
&\leq \frac{\bar{G}'\bar{\mu}^2}{\lambda\epsilon} +\left(\frac{\bar{G}'\bar{\mu}\rho}{\lambda(\mu^*-\epsilon\rho)^2} + \frac{\bar{G}'\bar{\mu}\bar{G}''}{2\lambda} + \frac{\bar{\mu}\bar{G}''+\bar{G}'}{2}\right)\cdot\left(\frac{m^2U^2\epsilon T}{2\lambda} + \frac{\bar{\mu}^2}{2\lambda\epsilon}+\frac{\bar{\mu}}{2\lambda\epsilon}\sum_{t=1}^T\gamma_t\right)\\
&\quad+ \frac{\bar{G}'\bar{\mu}}{2\lambda}\sum_{t=1}^T\gamma_t+ \sum_{t=1}^T\Delta_t+ JU\left(\frac{\bar{\mu}}{\epsilon\rho} + \frac{JU}{\rho}\right),
\end{aligned}
\end{equation*}
where the first inequality follows from Equation~\eqref{eq:upperboundrt} and the second inequality follows from Equation~\eqref{eq:upperbounddelta}.

In both cases, there exist constants $C_1,C_2,C_3,C_4$ such that
\begin{equation*}
\begin{aligned}
\text{Reg}^{\text{AVP}}(T,\rho) &\leq C_1\cdot\frac{1}{\epsilon} + C_2\cdot\epsilon T+C_3\frac{\sum_{t=1}^T\gamma_t}{\epsilon} + C_4\sum_{t=1}^T\Delta_t\\
&=\leq C_1\cdot\frac{1}{\epsilon} + C_2\cdot\epsilon T+C_3\frac{L'\sum_{j=1}^J\delta_{j,t}}{\epsilon} + C_4\sum_{t=1}^T\sum_{j=1}^J (UL+L)\delta_{j,t},
\end{aligned}
\end{equation*}
where the last line follows by substituting the bounds in step 2, Inequality~\eqref{eq:expenditurebound} and Inequality~\eqref{eq:utilitybound}. Assuming that $\sum_{t=1}^T\sum_{j=1}^J\delta_{j,t} = O(\sqrt{T})$, the regret is minimized by choosing $\epsilon = O(T^{-1/4})$, where we get $O(T^{3/4})$ regret.
\end{proof}

\begin{lemma}
\label{lemma:BR-estimate}
Suppose that for each auction $j$ and value $v$, the expected utility function $v\expalloc(b) - \exppay(b)$ is $\kappa$ strongly concave in $b$. For any $t\in[T],j\in[J]$, let $\hat{F}_{j,t}$ be the empirical bid distribution for auction $j$ at time $t$ in Algorithm~\ref{alg:online}, and $\hat{\sigma}_{j,t+1}:\mathbb{R}\rightarrow\mathbb{R}_{\geq0}$ be defined as $\hat{\sigma}_{j,t+1}(v)=v,\arg\max_{b\in\mathbb{R}}\E_{\hat{F}_{j,t}}[v\alloc(\bbf_{j,-(n+1)},b)- \pay(\bbf_{j,-(n+1)},b)]$. Moreover, recall that $\sigma_j(v) = \arg\max_{b\in\mathbb{R}}v\bar{x}_j(b) - \bar{p}_j(b)$, as defined in Equation~\eqref{def:bestresponse}. Then, for each round $t$, there exists $\delta_{j,t}>0$ such that $\sup_{v\in\mathbb{R}}|\sigma_j(v) - \hat{\sigma}_{j,t}(v)|\leq\delta_{j,t}$, where $\sum_{t=1}^T\delta_{j,t} = O(\sqrt{T})$ for each $j$.
\end{lemma}
\begin{proof}
Throughout the proof, we fix the auction $j$. From DKW inequality, at round $t$, because $\hat{F}_{j,t}$ has collected $n\cdot t$ bids from the first $t$ rounds, based on the Dvoretzky-Kiefer-Wolfowitz (DKW) inequality \cite{dvoretzky1956asymptotic, massart1990tight}, we have
\begin{align*}
\mathbb{P}\left(\sup_{x\in\mathbb{R}}|\hat{F}_{j,t}(x) - F_j(x)|>{\sqrt{\frac{\log nt}{2n}}}\right)
&\leq K\cdot \frac{1}{nt}
\end{align*}
for some constant $K>0$, which implies that there exists some constant $K'>0$ such that for any $v,b\in\mathbb{R}_{\geq 0}$,
\begin{align*}
\Bigg|\E_{\hat{F}_{j,t}}[v\alloc(\bbf_{j,-(n+1)},b)- \pay(\bbf_{j,-(n+1)},b)] - \E_{{F}_{j}}[v\alloc(\bbf_{j,-(n+1)},b)- \pay(\bbf_{j,-(n+1)},b)]\Bigg|\leq K'\cdot\frac{1}{nt},
\end{align*}
since the values are upper-bounded by a constant $U$. For simplicity, let $$\hat{h}(v,b) = \E_{\hat{F}_{j,t}}[v\alloc(\bbf_{j,-(n+1)},b)- \pay(\bbf_{j,-(n+1)},b)]$$ and $$h(v,b) = \E_{{F}_{j}}[v\alloc(\bbf_{j,-(n+1)},b)- \pay(\bbf_{j,-(n+1)},b)],$$ and let $\hat{\sigma}(v)$ and $\sigma(v)$ be the maximizers of $\hat{h}(v,b)$ and $h(v,b)$ over $b\in\mathbb{R}$ respectively. Since $h$ is concave,
we have 
$$h(v, \sigma(v))\geq h(v, \hat{\sigma}(v))\geq\hat{h}(v, \hat{\sigma}(v))-\frac{K'}{nt}\geq \hat{h}(v,  \sigma(v)) - \frac{K'}{nt}\geq h(v, \sigma(v)) - 2\frac{K'}{nt},$$
which implies
\begin{equation*}
    h(v,\sigma(v)) - h(v, \hat{\sigma}(v)) \leq 2\frac{K'}{nt}.
\end{equation*}

By strong concavity of $h$ we have
\begin{align*}
    &h(v, \hat{\sigma}(v)) \leq h(v, \sigma(v)) - h'(v, \sigma(v))(\hat{\sigma}(v)-\sigma(v)) + \frac{\kappa}{2}|| \hat{\sigma}(v)-\sigma(v)||^2\\
    &\Leftrightarrow\quad\quad | \hat{\sigma}(v)-\sigma(v)|\leq\sqrt{\frac{2}{\kappa}(h(\sigma(v))-h(\hat{\sigma}(v)))}\leq \frac{2K'}{\sqrt{n\kappa}}\cdot t^{-1/2}.
\end{align*}
Thus, we can set $\delta_{j,t} = \frac{2K'}{\sqrt{n\kappa}}\cdot t^{-1/2}$, where we have $\sum_{t=1}^T \delta_{j,t}\leq \int_{0}^{T}\frac{2K'}{\sqrt{n\kappa}}\cdot t^{-1/2} dt \leq \frac{2K'}{\sqrt{n\kappa}}\cdot \sqrt{T} = O(\sqrt{T})$.
\end{proof}

\end{document}